\documentclass{amsart}
\usepackage{bm}
\usepackage{verbatim}

\newcommand{\fo}{\mathfrak{o}}

\renewcommand{\d}{\mathrm{d}}

\newcommand{\End}{\operatorname{End}}

\renewcommand{\Re}{\operatorname{{Re}}}
\renewcommand{\Im}{\operatorname{{Im}}}
\newcommand{\img}{\operatorname{{img}}}
\newcommand{\rank}{\operatorname{{rank}}}
\newcommand{\spn}{\operatorname{{span}}}
\newcommand{\CH}{\mathrm{CH}}
\newcommand{\LHS}{\text{LHS}}

\newcommand{\SL}{\mathrm{SL}}
\newcommand{\SO}{\mathrm{SO}}
\newcommand{\so}{\mathfrak{so}}

\newcommand{\fg}{\mathfrak{g}}

\renewcommand{\S}{\mathbf{S}}
\newcommand{\Rset}{\mathbb{R}}

\newcommand{\Cset}{\mathbb{C}}

\newcommand{\bOmega}{\boldsymbol{\Omega}}
\newcommand{\be}{\boldsymbol{e}}
\newcommand{\bell}{\boldsymbol{\ell}}
\renewcommand{\bm}{\boldsymbol{m}}
\newcommand{\bn}{\boldsymbol{n}}
\newcommand{\bo}{\boldsymbol{o}}
\newcommand{\biota}{\boldsymbol{\iota}}
\newcommand{\bGamma}{\boldsymbol{\Gamma}}

\newcommand{\bomega}{\boldsymbol{\omega}}
\newcommand{\hbomega}{\hat{\bomega}}
\newcommand{\bA}{\boldsymbol{A}}

\newcommand{\tR}{\tilde{R}}
\newcommand{\tGamma}{\tilde{\Gamma}}
\newcommand{\hM}{\hat{M}}

\newcommand{\hR}{\hat{R}}
\newcommand{\homega}{\hat{\omega}}

\newcommand{\hGamma}{\hat{\Gamma}}
\newcommand{\tPsi}{\tilde{\Psi}}
\newcommand{\tPhi}{\tilde{\Phi}}
\newcommand{\tbGamma}{\tilde{\bGamma}}
\newcommand{\tbomega}{\tilde{\bomega}}

\newcommand{\cR}{\mathcal{R}}

\newtheorem{theorem}{Theorem}
\newtheorem{proposition}[theorem]{Proposition}

\theoremstyle{definition}
\newtheorem{definition}[theorem]{Definition}

\begin{document}

\title{Three-dimensional spacetimes of maximal order}
\author{R. Milson, L. Wylleman}
\email{robert.milson@dal.ca}
\email{lode.wylleman@ugent.be}
\address{Department of Mathematics and Statistics, Dalhousie University, Halifax, Canada}
\address{Department of Mathematical Analysis, Ghent University, B-9000 Ghent, Belgium}
\address{ Mathematical Institute, Department of Mathematics, Utrecht University, 3584 CD Utrecht, The Netherlands}
\address{Department of Mathematics and Natural Sciences, University of Stavanger,  N-4036 Stavanger,
Norway}
\begin{abstract}
  We show that the equivalence problem for three-dimensional
  Lorentzian manifolds requires at most the fifth covariant derivative
  of the curvature tensor.  We prove that this bound is sharp by
  exhibiting a class of 3D Lorentzian manifolds which realize this
  bound. The analysis is based on a three-dimensional analogue of the
  Newman-Penrose formalism, and spinorial classification of the
  three-dimensional Ricci tensor.
\end{abstract}
\maketitle
\section{Introduction and main result}

We report on recent progress concerning the invariant classification
problem for three-dimensional Lorentzian geometries.  In a physical
context,
such geometries arise as exact solutions of three-dimensional
theories of gravity, such as Topologically Massive Gravity (TMG),
New Massive Gravity (NMG) and
extensions of those. We refer to
\cite{cps10} and the introduction of \cite{aa12} for reviews of the
relevant literature. In \cite{cps10} it was stressed that, when
surveying the literature of exact solutions, it is often difficult
to disentangle genuinely new solutions from those that are already
known but written in different coordinate systems.

To tackle this problem one needs a 
{\em coordinate invariant} local characterization of the geometry. A
first step is to use the algebraic classification of the Ricci
tensor, as was done in \cite{cps10} to classify all TMG solutions
known at that time. A {\em complete} answer to the problem (in any
dimension in principle) is provided by the Cartan-Karlhede
algorithm~\cite{cartan,karlhede80}. The key quantities used here are
so-called {\em Cartan invariants}, which are components of the
Riemann tensor and a finite number of its covariant derivatives,
relative to some maximally fixed vector frame
associated to these tensors.

Regarding three-dimensional Lorentzian geometries, we will show in the
present paper that cases where one needs the theoretically maximal
number of five derivatives for a complete classification do exist, but
are limited to the metrics given in our main Theorem \ref{thm:main}
below. This implies that
\begin{quote}
  \emph{any three-dimensional geometric theory of gravity whose field
    equations exclude the metrics of Theorem \ref{thm:main} requires
    at most four covariant derivatives of the Riemann tensor for a
    complete local invariant classification of its exact solutions.}
\end{quote}
In the remainder of this introduction, we will outline the general
mathematical context and background for the main theorem.

Let $(M,g)$ be a smooth, $n$-dimensional pseudo-Riemannian manifold,
and let $(V,\eta)$ be a real inner-product space having the same
dimension and signature as $(M,g)$.  Henceforth, we use $\eta_{ab}$
to raise and lower frame indices, which we denote by
$a,b,c=1,\ldots, n$. Let $O(\eta)$ be the group of automorphisms of
$\eta$, and let $\fo(\eta)$ be
the corresponding Lie algebra of anti self-dual transformations.  An
$\eta$-orthogonal coframe is an inner-product isomorphism
\[ \omega_x: (T_x M,g_x) \to (V,\eta),\quad x\in M.\] Let
$\pi:O(\eta,M)\to M $ denote the principal $O(\eta)$-bundle of all
such.  An $\eta$-orthogonal \emph{moving coframe} is a local section
of this bundle, or equivalently, a collection of 1-forms $\omega^a$
such that
\[ g = \eta_{ab}\, \omega^a \omega^b.  \]
Set
\begin{equation}
  \label{eq:cRpdef}
  \cR^{p} = \otimes^4 V^*\oplus \cdots \oplus \otimes^{4+p}
  V^*
\end{equation}
and let $\hR^{(p)}:O(\eta,M)\to \cR^{p}$ be the canonical,
$O(\eta)$-equivariant map defined by
\begin{equation}
  \label{eq:hRpdef}
  \hR^{(p)} = (\hR_{abcd},\hR_{abcd;e},\ldots, \hR_{abcd;e_1\ldots
    e_p}),
\end{equation}
where the right hand side denotes the lift of the Riemann
curvature tensor and its first $p$ covariant derivatives to
$O(\eta,M)$.

The following definitions are adapted from \cite[Definitions 8.14 and
8.18]{olver}.  Set $r_{-1} = 0$, and let $r_p$ denote the rank of
$\hR^{(p)},\; p=0,1,2,\ldots$.  We say that $(M,g)$ is {\em fully
  regular} if $r_p$ is constant for all $p$. Henceforth we assume that
full regularity holds and let $q=q_M$ be the smallest integer such
that $r_{q-1} = r_q$.  The integer $q-1$ is called the \emph{order} of
the metric \cite{olver,ES}.  It can be shown\cite[Theorem
12.11]{olver} that a fully regular metric of order $q-1$ is classified
by $\hR^{(q)}$, that is by $q$th-order differential invariants.




The maximal order of a pseudo-Riemannian manifold, of fixed dimension
and signature, is of particular interest. Cartan \cite{cartan}
established the upper bound
\[ q\leq n(n+1)/2 = \dim O(\eta,M).\] Karlhede~\cite{karlhede80}
improved Cartan's bound
to
\begin{equation}
  \label{eq:kbound}
  q\leq n+ s_0+1,
\end{equation}
where $s_0$ is the dimension of the automorphism group of the
curvature tensor.
The question of maximal order has received considerable attention in
general relativity ($n=4$, Lorentzian signature)
\cite{coldinv93,MacCAman86,ramvick96}.
In that context, Karlhede's bound is $q\leq 7$; recently, this bound
was shown to be sharp~\cite{milpel08}. The 4-dimensional metrics of
maximal order describe a well-defined class of type $N$ spacetimes
with aligned null-radiation in an anti-deSitter
background~\cite{orr85}.
By contrast, Karlhede's bound in the generic Petrov type I case (for
which $s_0=0$) is $q\leq 5$, but at present we only have an example
of a type I dust solution \cite{wylleman08} with $q=3$.


In this paper, we investigate and classify 3-dimensional Lorentzian
manifolds of maximal order. Our approach is grounded in Karlhede's
refinement of the Cartan equivalence method~\cite{gardner}, which is
based on the notion of curvature normalization~\cite{karlhede80,ES}.
A non-zero three-dimensional curvature tensor has vanishing Weyl
part and is thus represented by its Ricci tensor, which may be
regarded as a
self-adjoint operator on the three-dimensional tangent space.
Generically, the Ricci operator has a finite automorphism group
(whence $s_0=0$).
However, if two eigenvalues coincide or if the trace-free part of
the operator is nilpotent, then $s_0=1$ is possible. Therefore, in
the three-dimensional Lorentzian setting, Karlhede's bound is $q\leq
5$ \cite{sfr08}.  The question then becomes:
\begin{quote}
  \it Does there exist a 4th order, 3-dimensional Lorentzian
  metric, that is to say, a metric that is classified by 5th-order
  differential invariants?
\end{quote}

In 3-dimensional Lorentzian geometry, it is useful to make use of the
real spinor representation of the Lorentz group. Such a spinor
approach provides one with a natural null vector frame formalism.
Moreover, the Petrov-Penrose classification of the curvature spinor
(which, in three dimensions, is equivalent to the null alignment
classification of the Ricci tensor) leads to a slight refinement of
the usual Ricci-Segre classification. This is summarized in the
appendices.

Karlhede's result, which we formulate as Theorem \ref{thm:karlhede}
below, tells us that a metric which is classified by 5th order
invariants, if one exists, is restricted to Petrov type D,
type DZ (like type D, but the doubly aligned null directions are
complex) and type N
geometries.
Below,
we rule out the type DZ and N possibilities, and demonstrate that
the $q=5$ bound is realized for one very particular class of type D
metrics.

\begin{theorem}
  \label{thm:main}
  The order of a curvature-regular, 3-dimensional Lorentzian manifold
  is bounded by
  \[ q-1 \leq 4.\] This bound is sharp; every 4th order
  metric is locally isometric to
  \begin{align}
    \label{eq:ds1}
    &2 (2 T x du + dw) ^2 -  2 du (dx+a du) ,\quad\text{where}\\
    \label{eq:adef}
    & a= \frac{1-e^{4Tw}}{2T} + (2T^2-C) (x-\delta_C)^2 + F(u).
  \end{align}
  Here $x,u,w$ are local coordinates. $C,T$ are real constants such
  that $C+2T^2\neq 0$, and $F(u)$ is an arbitrary real function such
  that
  \begin{equation}
    \label{eq:Fu-maxorder-cond}
    \begin{cases}
      (1+2TF(u))F''(u)\neq 3T(F'(u))^2  & \text{if } C \neq 0,\\
     F'(u)\neq 0 & \text{if } C = 0.
    \end{cases}
  \end{equation}
\end{theorem}
\noindent Note 1:
In the singular subcase of $T=0$, the expression
$(1-e^{4Tw})/(2T)$ should be interpreted in the limit sense as being
equal to $-2w$. \\
Note 2: the expression $\delta_C$ denotes 1 if $C=0$ and 0 if
$C\neq 0$.
\\
\noindent Note 3: curvature regularity
is a strengthening of the full-regularity assumption that we impose
in order to exclude ``type-changing'' metrics (see Definition
\ref{def:curvreg} below).

The structure of this paper is as follows. In Section
\ref{sec:curvnorm} we revise the relevant definitions and theorems
regarding curvature normalization, leading to Karlhede's bound
within his approach to the equivalence problem.
The concepts of curvature homogeneity and pseudo-stabilization turn
out to be the crucial ideas in the search for metrics of maximal
order. In particular, the maximal order metrics shown in
\eqref{eq:ds1} enjoy the $\CH_1$ (curvature homogeneous of order 1)
property. The relevant definitions are given in Section
\ref{sec:curvhom}.  We isolate the structure equations for the
maximal order metrics in Section \ref{sec:equivprob}. We then prove
the main Theorem \ref{thm:main} by integrating these equations in
Section \ref{sec:3Dmax}.
Relevant background material is put in four appendices: a
three-dimensional analogue of the Newman-Penrose formalism, the
transformation rules of connection and curvature variables under basic
Lorentz transformations, the Petrov-Penrose
classification of the three-dimensional Ricci tensor, and the
structure equations obeyed by a $\CH_1$ metric.

\section{Curvature normalization and Karlhede's
  bound}
\label{sec:curvnorm}
A general approach towards finding
metrics of maximal order was described in \cite{coldinv93} and
\cite{milpel09}.  The approach is based on two key ideas: (i)
curvature normalization, also known as the Karlhede
algorithm~\cite{karlhede80}, and (ii) curvature
homogeneity~\cite{singer}.
Normalization of the curvature tensor and its covariant derivatives,
also known as the Karlhede algorithm, splits the rank of the
classifying map $\hR^{(p)}$ into horizontal and vertical subranks and
thereby simplifies the equivalence problem. As was already mentioned,
the rank $r_p$ is the maximal number of functionally independent
component functions $(\hR_{abcd},\ldots ,\hR_{abcd;e_1\ldots e_p})$,
where the latter are functions of both position and frame
variables. In order to speak of horizontal rank, we need to assume
that the above tensors can be normalized.  The horizontal rank (see
Definition \ref{def:hrank} below) can then be defined as the the
maximal number of functionally independent component functions of
normalized curvature and its covariant derivatives.
\begin{definition}
  \label{def:curvreg}
  We say that a submanifold $S\subset \cR^{p}$ is a $p$th order
  \emph{normalizing cross-section} for $(M,g)$ provided:
\begin{itemize}
\item[(N1)] there exists a subgroup $G_p\subset O(\eta)$ that fixes $S$
  pointwise;
\item[(N2)] the normalization is maximal in the sense that
  $X(S)\cap S\neq \emptyset,\; X\in O(\eta)$ implies $X\in G_p$;
\item[(N3)] $(M,g)$ admits a cover by $\eta$-orthonormal moving
  coframes such that
  \[\img R^{(p)}\subset S,\quad \text{where }
  R^{(p)} = (R_{abcd},\ldots ,R_{abcd;e_1\ldots e_p})\] denotes the
  curvature components relative to the coframe in question.
\end{itemize}
If there exists a normalizing cross-section $S\subset \cR^{p}$ for
every $p=0,1,2,\ldots$ we say that $(M,g)$ is \emph{curvature
  regular}.
\end{definition}

Suppose that curvature regularity holds.  Normalizing $R^{(p)}$
reduces the structure group of the equivalence problem from $O(\eta)$
to $G_p$.  Because of N2, the maximally normalized components
$(R_{abcd}, \ldots, R_{abcd;e_1\ldots e_p})$ are locally defined
functions on the base $M$. These differential invariants, commonly
referred to as $p$th order \emph{Cartan invariants}, suffice to
invariantly classify $(M,g)$ and to solve the metric equivalence
problem \cite[Chapter 9]{ES}.

\begin{definition}
\label{def:hrank} Suppose that $(M,g)$ is curvature regular. We
define
\begin{align}
  \label{eq:nupdef}
  s_p &:= \dim G_p,\\
  \label{eq:taupdef}
   t_p & := \rank R^{(p)}
\end{align}
relative to some choice of normalizing cross-section. We refer to
$s_p$ as the $p$th order degree of frame freedom, and to $t_p$ as the
$p$th order horizontal rank.
\end{definition}
\begin{proposition}
  If $(M,g)$ is curvature regular, then $s_p,t_p$ do not vary
  with $x\in M$ and are independent of the choice of normalizing
  cross-section.  Furthermore,
  \begin{equation}
    \label{sp and tp variation}
    s_{p}\leq s_{p-1},\quad t_{p}\geq t_{p-1}
  \end{equation}
  and
  \begin{equation}\label{eq:rpsptp}
  r_p = t_p + n(n-1)/2-s_p,\quad p=0,1,2,\ldots.
  \end{equation}
\end{proposition}
\begin{theorem}[Karlhede, Theorem 4.1 of \cite{karlhede80}, see also
  Section 9.2 of \cite{ES}]
  \label{thm:karlhede}
  Let $(M,g)$ be a fully regular, curvature regular $n$-dimensional
  pseudo-Riemannian manifold with isometry group $K$. Let $r_p, t_p,
  s_p$ be as defined above, and let $q$ be the smallest integer such
  that $r_{q-1} = r_q$.  Then, $q$ is also the smallest integer such
  that $s_{q-1} = s_q$ and $t_{q-1} = t_q$. Furthermore, we have that
   $G_{q-1}\subset O(\eta)$ is isomorphic to the isotropy subgroups
  $K_x\subset O(T_xM),\, x\in M$; that
  \begin{equation}\label{eq:dimK}
  \dim K = n-t_{q}+s_{q},
  \end{equation}
  and that $n-t_{q}$ is equal to the dimension of the $K$-orbits.
\end{theorem}
\noindent
In particular, \eqref{eq:rpsptp} implies that
\begin{equation}
  \label{eq:r0s0}
   r_0 \geq n(n-1)/2 - s_0.
\end{equation}
By the regularity assumption,
\[ r_0 + p\leq r_p \leq n(n+1)/2 ,\quad 0\leq p\leq q-1.\] Applying
the above inequality with $p=q-1$ and using \eqref{eq:r0s0} gives the
Karlhede bound \eqref{eq:kbound} as an immediate corollary.



\section{Curvature homogeneity and  pseudo-stabilization}
\label{sec:curvhom}

Suppose that $(M,g)$ is fully regular and curvature regular.  The
curvature-homogeneity  condition admits several equivalent
definitions~\cite{boeckx,gilkey}, but with the above assumptions,
the following definition is the most convenient.
\begin{definition}\label{def:curv-hom}
  A manifold $(M,g)$ is {\em curvature-homogeneous} of order $k$, or
  $\CH_k$ for short, if it is curvature regular and if the horizontal
  rank $t_k=0$.
 If $t_k=0$ and $t_{k+1}>0$, we say that $(M,g)$ is
  \emph{properly} $\CH_k$.
\end{definition}
\noindent
To put it another way, a properly curvature
homogeneous manifold of order $k$ has constant Cartan invariants of
order $\leq k$, with a non-constant invariant appearing at order
$k+1$.
The main application of the curvature homogeneous concept was the
following theorem \cite{singer}.
\begin{theorem}[Singer]
  A manifold $(M,g)$ is locally homogeneous if and only if it is
  $\CH_k$ for all $k=0,1,2,\ldots$.
\end{theorem}
\noindent In other words, a locally homogeneous space is
characterized by the property of having constant Cartan invariants.
As such, Singer's theorem is an immediate corollary of Theorem
\ref{thm:karlhede}.

In this paper we are interested in curvature homogeneity
for a different, but related reason. As was shown in
\cite{milpel08}, curvature homogeneity is also a key concept in the
search for maximal order metrics.
The relevant observation is that for a
$\CH_k$ geometry the rank $r_k$ is small because $t_k=0$,
and this is exactly what is needed for maximal
order.
Let us explain further in the context of 3-dimensional Lorentzian
metrics.


\begin{definition}\label{def:pseudo-stab}
  We say that a curvature regular geometry has $k$th order
  pseudo-stabilization provided $s_k=s_{k-1}>s_q$.
\end{definition}
\noindent Our notion of pseudo-stabilization is different but
conceptually related to the notion employed in \cite[Theorem
5.37]{olver}. Notice that a $k$th order pseudo-stable geometry has
$t_k> t_{k-1}$ by theorem \ref{thm:karlhede}.


\begin{proposition}
  \label{prop:pseudostab}
  A 4th order, 3-dimensional, Lorentz geometry, if one exists, is
  either properly $\CH_1$ or is properly $\CH_0$  with 1st order
  pseudo-stabilization.
\end{proposition}
\begin{proof}
  Table \ref{tab:petrovtype} reveals $s_0 \leq 1$ for a
  non-homogeneous geometry (see Proposition \ref{prop:schur} below).
  Hence, $r_0\geq 2$, and hence a 4th order geometry requires the
  following rank sequence:
  \[ (r_p) = (2,3,4,5,6,6).\]
  This can be achieved in essentially two ways: either by
  \begin{equation}
    \label{eq:CH1maxorder}
    (t_p) = (0,0,1,2,3,3),\quad (s_p) = (1,0,0,0,0,0),
  \end{equation}
  which describes a properly $\CH_1$ geometry, or by by three possible sequences
with $s_q = 0$ and starting with
  \begin{equation}
    \label{eq:CH0ps1maxorder}
     (t_p) = (0,1,\ldots),\quad (s_p) = (1,1,\ldots),
  \end{equation}
  which describes a properly $\CH_0$ geometry with 1st order
  pseudo-stabilization.
\end{proof}
In the following section, we rule out the pseudo-stabilization and
type DZ, N scenarios and show that a 4th order requires a type D,
properly $\CH_1$ geometry.  We then explicitly write down the
necessary structure equations and integrate them. The end result is
Theorem \ref{thm:main}.


\section{The equivalence problem}
\label{sec:equivprob}
In this section we derive the necessary and sufficient conditions for
a 4th order metric.  Table \ref{tab:petrovtype} of the appendix shows
that $s_0 >0$ for curvature types O, N, DZ, and D.  Type O can be
ruled out by Schur's theorem.  A proof can be found in
\cite[Cor. 2.2.5 and 2.2.7]{wolf}.
\begin{proposition}
  \label{prop:schur}
  If the curvature is type O at all points $x\in M$, then $M$ is a
  locally homogeneous space, i.e. $t_p=0$ for all $p$.
\end{proposition}
\noindent
We are
left with the following possibilities.
\begin{proposition}
  \label{prop:NDZD}
  A 4th order metric, if one exists, requires curvature of type N, DZ,
  or D.
\end{proposition}

According to Proposition \ref{prop:pseudostab}, each of the above 3
cases further splits into two subcases, according to whether the
geometry is properly $\CH_1$ or properly $\CH_0$ with 1st order
pseudo-stabilization.  We consider the above possibilities in turn.
Five of the possibilities can be ruled out, and this leaves a unique
configuration for a 4th order metric.

Since in a $\CH_0$ geometry the 0th order components $R_{abcd}$ are
constant, the 1st order components $R_{abcd;e}$ are quadratic
expressions of certain spin coefficients.  Therefore, in the
analysis that follows it is more convenient to specify the Cartan
invariants in terms of spin coefficients and their frame
derivatives.  This methodology for constructing invariants is
related to the notion of essential torsion in the Cartan equivalence
method. The relevant details and definitions are given in Appendix
\ref{sec:ch1}.


\subsection{Type N configurations.}\label{subsec:typeN}
Taking the curvature canonical form of Table \ref{tab:petrovtype} for
this case, and assuming the $\CH_0$ property, we have
\begin{equation}\label{eq:Nconfig0}
  \Psi_0=\Psi_1=\Psi_2 = \Psi_3=0,\quad \Psi_4=\pm 1,\quad
  R=\tilde{R},
 \end{equation}
 where $\tilde{R}$ is a real constant.
The group $G_0$ preserving \eqref{eq:Nconfig0} is generated by null
rotations \eqref{eq:nullrot} about $\bell$ and the reflections
\eqref{eq:reflt_def}, \eqref{eq:reflm_def}.
The type N 1st order torsion matrix (see Appendix \ref{sec:ch1} for
the derivation) is
\begin{equation}\label{type N Gamma matrix} (\Gamma^\rho{}_a) =
  \begin{pmatrix}
    \kappa & \sigma & \tau \\
    \epsilon & \alpha & \gamma
  \end{pmatrix}.
\end{equation}
Substituting \eqref{eq:Nconfig0} into the Bianchi equations
\eqref{eq:bianchi1}-\eqref{eq:bianchi3} yields the relations
\begin{equation}
  \label{eq:Nconfig1}
  \kappa=0,\quad \sigma=2\epsilon.
\end{equation}
We now consider the $\CH_1$ and pseudo-stable cases in turn.

\begin{proposition}
  \label{prop:NCH1}
  A type N, properly $\CH_1$ geometry has order bounded by $q-1 \leq
  3$.
\end{proposition}
\begin{proof}
  By assumption, after the first-order torsion is normalized,
  $\epsilon,\tau,\alpha,\gamma$ are constants.
  Hence, by \eqref{eq:NP1} - \eqref{eq:NP5}
  \begin{eqnarray}
    &&\label{eq:Nconfig2}
    \sigma = \epsilon = 0,\quad R = -12 \tau^2,\\
    &&\label{eq:Nconfig3}
    (\tau+\pi)(2\alpha+\tau) =0.
  \end{eqnarray}
  By \eqref{eq:nr1}-\eqref{eq:nr6}, $\tau$ and $\alpha$ are invariant
  under any null rotation about $\bell$, while $\gamma$ transforms like
  \begin{equation}
    \label{eq:Nconfig4}
    \gamma' = \gamma+x(2\alpha+\tau).
  \end{equation}
  Since $t_0=t_1=0$ and $s_0=1$ by assumption, $s_1=1$ would lead to $q-1=0$. Thus we assume $s_1=0$ henceforth. By \eqref{eq:Nconfig4} this entails $2\alpha+\tau\neq 0$, and hence
  \begin{equation}
    \label{eq:Nconfig5}
    \pi = -\tau
  \end{equation}
  by \eqref{eq:Nconfig3}.  We impose the normalizations
  \begin{equation}
    \label{eq:Nconfig6}
    \gamma = 0,\quad 2\alpha+\tau>0,
  \end{equation}
  which leaves $G_1$ as the discrete group generated by \eqref{eq:reflt_def}.
  Then equation \eqref{eq:NP7} implies that
  \begin{equation}
    \label{eq:Nconfig7}
    \lambda =0.
  \end{equation}
From \eqref{eq:Nconfig5}, \eqref{eq:Nconfig7} and $s_1=0$ it follows
that the 2nd order invariants are generated by $\nu$. If $\nu$ is
constant then $t_2=0$ and $q-1=1$. Thus we assume henceforth that
$\nu$ is non-constant, i.e.\ $t_2=1$.
The remaining structure equations \eqref{eq:NP8} and \eqref{eq:NP9}
reduce to
  \begin{equation}
    \label{eq:Nconfig8}
    D\nu = 0,\quad \delta\nu= 2\nu(\tilde\tau-2\tilde\alpha)\pm 1/2,
  \end{equation}
  where $\tilde\tau, \tilde\alpha$ are constants such that
  $2\tilde\alpha+\tilde\tau> 0$.
  Suppose then
  that $\Delta\nu$ is functionally independent from $\nu$, and hence
  that $t_3=2$ (else $t_3=1$ and $q-1=2$). By the (N2) curvature
  regularity assumption we have $\Delta\nu \neq 0$ at each point and we fully fix the frame by normalizing
  \[ \Delta \nu > 0.\]
  Now the 3rd order invariants are generated by $\nu, \Delta \nu$.
  Applying \eqref{eq:commutator2} and \eqref{eq:commutator3} to $\nu$
  gives
  \begin{equation}\label{eq:Nconfig8b}
    D\Delta\nu = 0,\quad \delta\Delta\nu = 3 (\tilde\tau-2\tilde\alpha)
    \Delta \nu.
  \end{equation}
  Hence, the 4th order invariants are generated by
  $\nu,\Delta\nu$ and $\Delta^2\nu$. Applying \eqref{eq:commutator2}
  to $\Delta\nu$ gives
  \[D\Delta^2\nu = 0,\] and hence
  \[ \d\nu \wedge \d\Delta\nu \wedge \d \Delta^2\nu = 0.\]
  Therefore $t_4=2$, which implies that the order is $q-1=3$.
\end{proof}



\begin{proposition}
  \label{prop:NCH0}
  The order of a type N, $\CH_0$, pseudo-stable geometry is bounded by
  $q-1\leq 3$.
\end{proposition}
\begin{proof}
  Referring to \eqref{type N Gamma matrix}, the assumption $s_1=1$
  implies that the remaining torsion scalars $\epsilon,\tau,\alpha$
  and $\gamma$ are invariant under arbitrary null rotations about
  $\bell$ and thus generate the 1st order Cartan invariants. By
  \eqref{eq:nr3}-\eqref{eq:nr6} this implies
  \[ \epsilon=\sigma = 0,\quad \alpha=-\tau/2.\] Hence the 1st order
  invariants $R_{abcd;e}$ are generated by $\tau, \gamma$. By the
  pseudo-stabilization assumption $R_{abcd;e}$ is $G_0$-invariant, and
  hence, using the notation of Appendix \ref{sec:ch1},
  \[ R_{abcd;ef} = (\nabla R)_{abcde,f} + \Gamma^\rho{}_f (\bA_\rho
  \cdot \nabla R)_{abcdef}.\] It follows that the 2nd order components
  are linear combinations of $D\gamma,D\tau,\delta
  \tau,\Delta\tau,\delta\gamma,\Delta\gamma$ and quadratic polynomials
  of $\tau,\gamma$. Since the latter are null-rotation invariant, the
  2nd order Cartan invariants are obtained by normalizing the former.

  From here equations \eqref{eq:NP1}-\eqref{eq:NP5} reduce to
  \begin{gather}
    \label{eq:Nps1}
    D\tau = 0,\quad \delta\tau=-D\gamma = \tR/12+\tau^2.
  \end{gather}
  If $t_1>1$ then $q-1\leq 3$ automatically, so we may assume $t_1=1$. This implies $\d\tau \wedge \d\gamma = 0$, and in
  particular $D\gamma\delta\tau=0$.  Hence,
  \begin{equation}
    \label{eq:Nps2}
    D\gamma =  \tR+12\tau^2=0,
  \end{equation}
  which implies that $\tau = \tilde\tau$ is a constant.  This leaves
  $\gamma$ as the only generator of the 1st order invariants.  The
  transformation law \eqref{eq:nullrot} gives
  \begin{equation}\label{Deltagamma_trafo}
    \Delta'\gamma = \Delta\gamma -4x\gamma\tilde\tau.
  \end{equation}
  At this point, we must consider two cases.

  Case (a): suppose that $\tilde\tau=0$ Hence, $\Delta\gamma$ is
  null-rotation invariant, and hence is a Cartan invariant.  By
  \eqref{eq:NP7}
  \[ \delta \gamma = 0 .\]
  Since $\gamma,\Delta \gamma$ generate the 2nd order invariants, we
  have $s_2=1$. Applying
  \eqref{eq:commutator2} and \eqref{eq:commutator3} to $\gamma$ gives
  \begin{equation}
    \label{eq:Nps5}
    D\Delta\gamma = \delta\Delta\gamma= 0
  \end{equation}
  Hence, $\d\gamma \wedge \d\Delta\gamma = 0$. This implies $t_2=1$
  and thus $q=2$. Hence, the corresponding geometries
are not pseudo-stable.

  Case (b): suppose that $\tau\neq 0$. In view of
  \eqref{Deltagamma_trafo}, \eqref{eq:reflt_conn} and
  \eqref{eq:reflm_conn} we may fully fix the frame ($s_2=0$) by
  imposing the normalizations
  \begin{equation}
    \label{eq:Nps6}
    \gamma>0,\quad \tau>0,\quad \Delta\gamma = 0.
  \end{equation}
  By \eqref{eq:NP7},
  \begin{equation}
    \label{Dgamma_trafo}
    \delta\gamma=2\tilde\tau\gamma.
  \end{equation}
  It follows that $t_2=1$, and that the 1st and 2nd order invariants
  are generated by $\gamma$.  Again, using the notation of Appendix
  \ref{sec:ch1},
  \begin{equation}
    \label{eq:R3}
    R_{abcd;e_1e_2f} = (\nabla^2R)_{abcde_1e_2,f} + \Gamma^\alpha{}_f
    (\bA_\alpha\cdot \nabla^2 R)_{abcde_1e_2}.
  \end{equation}
  By \eqref{eq:Nps2} \eqref{eq:Nps6} \eqref{Dgamma_trafo}, the
  components $(\nabla^2R)_{abcde_1e_2,f}$ are generated by
  $\gamma$. Since the automorphism group of $\nabla^2 R$ is trivial,
  equations \eqref{eq:R3} can be solved for $\Gamma^\alpha{}_f$.  It
  follows that $\lambda,\nu,\pi$, together with $\gamma$, generate
  invariants of order 3 or less.  The commutator relations
  \eqref{eq:commutator2} and \eqref{eq:commutator3} applied to
  $\gamma$ give
  \[ \gamma\tilde\tau(\pi+\tilde\tau) = 0,\quad \gamma\tilde\tau\lambda=0.\]
  Hence,
  \begin{equation}
    \pi=-\tilde\tau,\quad\lambda =0.
  \end{equation}
  The remaining structure equations \eqref{eq:NP8} and \eqref{eq:NP9} reduce to
  \begin{gather}
    D\nu = 0,\quad \delta\nu = 4\nu\tilde\tau\pm 1/2.
  \end{gather}
  Observe that $d\gamma\wedge d\nu =0$ if and only if $\Delta\nu=0$;
  the corresponding rank sequence is
  \[ (s_p) = (1,1,0,0),\quad (t_p) = (0,1,1,1), \]
  and the order is $q-1=2$. Else we have $t_3=2$, the 4th order invariants
  being generated by $\gamma,\nu$ and $\Delta\nu$.
  Applying \eqref{eq:commutator2} to $\nu$ gives
  $ D\Delta\nu = 0$
  and thus
  \[ \d\gamma\wedge \d\nu\wedge \d\Delta \nu = 0.\]
  Hence, the rank sequence is
  \[ (s_p) = (1,1,0,0,0),\quad (t_p) = (0,1,1,2,2), \] and the order
  is equal to $q-1=3$.
\end{proof}

\subsection{Type DZ configurations}\label{subsec:typeDZ}
For this case, combining the curvature canonical form of Table
\ref{tab:petrovtype} and the $\CH_0$ property $t_0=0$ gives
\begin{equation}
  \label{eq:DZ1}
  \Psi_1=\Psi_3= 0,\quad \Psi_0=\Psi_4= 3\Psi_2=3\tilde{\Psi}_2\neq 0,\quad
  R=\tR,
\end{equation}
where $\tilde\Psi_2$
and $\tilde{R}\in\Rset$ are real constants.
The group $G_0$ preserving \eqref{eq:DZ1}
is generated by  spins \eqref{eq:spin} and the reflections \eqref{eq:reflt_def},
\eqref{eq:reflm_def}, \eqref{eq:reflx_def}.
The 1st order torsion is
\begin{equation}
  \label{eq:DZ2}
  (\Gamma^\rho{}_a) =
  \begin{pmatrix}
    \epsilon & \alpha & \gamma\\
    \kappa-\pi & \sigma-\lambda & \tau-\nu
  \end{pmatrix}.
\end{equation}
Substituting \eqref{eq:DZ1} into the Bianchi
equations \eqref{eq:bianchi1}-\eqref{eq:bianchi3} yields
\begin{equation}
   \label{eq:DZ3}
   \sigma-\lambda = -2\gamma = 2\epsilon,\quad  \kappa- \pi = -(\tau-\nu).
\end{equation}


\begin{proposition}
  \label{prop:DZCH0}
  There does not exist a type DZ, $\CH_0$ geometry with
  pseudo-stabilization.
\end{proposition}
\begin{proof}
  The assumption implies that $t_0=0,t_1>0$ and that the 1st order torsion
  is spin-invariant.  By \eqref{eq:spin4}, this entails
  \begin{gather}
    \sigma - \lambda = \gamma=\epsilon = 0,\quad \kappa-\pi
    =-(\tau-\nu)= 2\alpha.
  \end{gather}
  Hence, $\alpha$ generates the 1st order invariants.  However,
  \eqref{eq:NP5} entails
  \begin{equation}
    0=2 ( \tau+ \kappa-2\alpha) \alpha +(\kappa-2\alpha) \tau  -\kappa
    (2\alpha+\tau)  +\Psi_2-R/12 = -4\alpha^2+\tilde\Psi_2-\tilde R/12.
  \end{equation}
  This implies that $\alpha$ is a constant, which contradicts the
  $t_1>0$ assumption.
\end{proof}

\begin{proposition}
  \label{prop:DZCH1} A type DZ, properly $\CH_1$ geometry does not
  exist.
\end{proposition}
\begin{proof}
  In addition to \eqref{eq:DZ1} we assume that $t_1=0,\; t_2>0$
  and that $s_1=0$.  The $t_1=0$ assumption means that
  post-normalization, the torsion components \eqref{eq:DZ2} are
  constant, say
\begin{equation}
  \label{eq:DZ2a}
  -(\kappa-\pi) = \tau-\nu = C_1,\quad \sigma-\lambda = -2\gamma = 2\epsilon = C_2,\quad \alpha = \tilde\alpha .
\end{equation}
Transformation law \eqref{eq:spin4} now reads
\[ (C_1 + 2\alpha + 2 i C_2)' = e^{\pm 2i t} \text{LHS}.\]
Since $s_1=0$ this  cannot be zero and we may therefore impose the
normalization
\[ C_2=0,\quad C_1+2\tilde\alpha>0.\] Applying \eqref{eq:DZ2a},
\eqref{eq:DZ3} to equation \eqref{eq:NP3} gives
\begin{equation}
  \label{eq:DZ6}
  \sigma(C_1+2\tilde\alpha)=0
\end{equation}
Then, \eqref{eq:NP1}-\eqref{eq:NP9} entail
\[ \sigma = \lambda = 0,\quad\pi = \tau=-\nu=-\kappa=C_1/2. \]
Hence, all 2nd order Cartan invariants are constant, a
contradiction.
\end{proof}

\subsection{Type D configurations}\label{subsec:typeD}
\begin{proposition}
  \label{prop:DCH0}
  There does not exist a type D, properly $\CH_0$ geometry with
  pseudo-stabilization.
\end{proposition}
\begin{proof}
  The curvature normalization from Table \eqref{tab:petrovtype} and the $\CH_0$ assumption give
  \begin{equation}
    \label{eq:DD1}
    \Psi_0=\Psi_1=\Psi_3=\Psi_4 = 0,\quad \Psi_2=\tilde\Psi_2,\;
    R=\tilde{R},
  \end{equation}
  where $\tilde\Psi_2\neq 0$ and $\tilde{R}$ are constants, and we
  also assume
  \[ s_0=s_1=1,\quad t_0=0,\; t_1> 0. \]
  The 1st order torsion is
  \begin{equation}
    \label{eq:typeDtorsion}
    (\Gamma^\rho{}_a)=
    \begin{pmatrix}
      \kappa & \sigma & \tau\\
      \pi &\lambda & \nu
    \end{pmatrix},
  \end{equation}
  where, by the Bianchi relations,
  \begin{equation}
    \label{eq:Dconfig1}
    \sigma=\lambda=0,\quad \pi=\tau.
  \end{equation}
  By the boost transformation laws \eqref{eq:boost1}-\eqref{eq:boost6}, in order to have $s_1=1$ we require
  \[ \kappa = \nu = 0.\] Adding \eqref{eq:NP4} to \eqref{eq:NP6}
  yields
  \[0=2\tau^2+\tilde{\Psi}_2-\tR/6,\] which implies that $\tau$ is
  constant.  Hence $t_1 = 0$, contradicting our assumption.
\end{proof}



\begin{proposition}
  \label{prop:DCH1}
  Up to $O(\eta)$ conjugation, the unique type D, properly $\CH_1$
  configuration is
  \begin{gather}
  \label{eq:D1}
  \Psi_0=\Psi_1=\Psi_3=\Psi_4 = 0,\quad \Psi_2=-\frac 23
  (C+2T^2),\quad R=4(C-T^2),\\
    \label{eq:D13}
    \sigma=\alpha=\gamma=\lambda=\nu=0,\quad
    \pi=\tau=T,\quad \kappa=1,\\
    \label{eq:D8}
    \delta \epsilon= -T\epsilon,\quad
  \Delta\epsilon = C,\quad \epsilon > 0,
  \end{gather}
  where $T$ and $C$ are constants such that $C+2T^2\neq 0$.
\end{proposition}
\begin{proof}
Suppose that the curvature is type D, and that
\[ (t_p) = (0,0,t_2,\ldots),\quad t_2>0 .\]
As above we have \eqref{eq:DD1} and \eqref{eq:typeDtorsion}. The
curvature automorphism group $G_0$ is generated by the 1-dimensional
group of boost transformations \eqref{eq:boost0} and the discrete
transformations \eqref{eq:reflt_def}, \eqref{eq:reflm_def},
\eqref{eq:reflx_def}. The corresponding transformation laws are
shown in Appendix \ref{app:Lorentz-trafos}.
Since $t_0=t_1=0$ this means that post-normalization, $R, \Psi_2,
\kappa, \sigma, \tau, \pi, \lambda,\nu$ are all constant, where we
put $\pi=\tau=T$. As above, the Bianchi identities give
\eqref{eq:Dconfig1}. Equations \eqref{eq:NP1}, \eqref{eq:NP2},
\eqref{eq:NP8} and \eqref{eq:NP9} reduce to
\begin{equation}
  \label{eq:D4}
  \alpha\kappa = \gamma\kappa = \epsilon \nu = \alpha\nu = 0.
\end{equation}
If $\kappa= \nu\equiv0$, identically, then by \eqref{eq:boost1} -
\eqref{eq:boost6} the first order torsion is boost-invariant, which
violates the assumption $s_1=0$.  Suppose then that $\kappa\neq 0$.
By the (N2) maximality assumption in Definition \ref{def:curvreg},
$\kappa$ cannot change sign.  Using \eqref{eq:boost5} and
\eqref{eq:reflm_conn}, we impose the normalization
\begin{equation}
  \label{eq:D6}
  \kappa = 1.
\end{equation}
The 2nd order torsion is $\alpha,\gamma,\epsilon$. If $\nu\neq 0$
then, by \eqref{eq:D4} the 2nd order torsion vanishes, which
violates the assumption $t_2>0$.   Therefore,
\begin{equation}
  \label{eq:D5}
  \nu = \alpha=\gamma\equiv 0,
\end{equation}
identically. Hence, $t_2=1$. The case where $\kappa\equiv 0, \nu\neq
0$ does not not need to be analyzed, because it can be reduced to
the present case by the Lorentz transformation \eqref{eq:reflx_def},
\eqref{eq:reflx_conn}. Again by the (N2) assumption of maximal
normalization, $\epsilon$ must have definite sign. Using
\eqref{eq:reflt_def}, \eqref{eq:reflt_conn} to impose the
normalization $\epsilon>0$ fully fixes the frame. Taking the second
part of \eqref{eq:D1} as a definition for the constant $C$, equation
\eqref{eq:NP4} gives
\[
  R =  -6\Psi_2-12\tau^2 = 4(C-T^2).
\]
The rest
of \eqref{eq:NP1}-\eqref{eq:NP9} are either satisfied identically, or
reduce to \eqref{eq:D8}.
\end{proof}

Above, we have derived a unique set of necessary conditions for a
type D properly $\CH_1$ geometry.  In other words, if such a metric
exists, then around every point there exists a unique
null-orthogonal moving frame such that \eqref{eq:D1} --
\eqref{eq:D8} hold.  Such geometries feature 1st order invariants
$C,T$, which must be constants, and a unique, up to functional
dependence, non-constant 2nd order invariant $\epsilon$.  This is
the necessity question.  Next, we consider sufficiency.

The configuration equations \eqref{eq:D1} -- \eqref{eq:D8} constitute
a system of partial differential equations for type D, properly
$\CH_1$ metrics.  We reformulate this system as the structure
equations of a generalized Cartan realization problem
\cite[appendix]{bryant2} \cite[Section 3]{fernstru} using Bryant's
recent treatment \cite{bryant1} of the realization problem.  To wit,
\eqref{eq:D1}  -- \eqref{eq:D8} is equivalent to
\begin{align}
  \label{eq:dom0}
  &\d\omega^0 = -T \omega^0\wedge\omega^1,\\
  \label{eq:dom1}
  &\d\omega^1 = -4T \omega^0\wedge \omega^2,\\
  \label{eq:dom2}
  &\d\omega^2 = \omega^0\wedge\omega^1+2\epsilon\omega^0\wedge\omega^2
  - T \omega^1\wedge\omega^2,\\
  \label{eq:depsilon0}
  &\d\epsilon = P \omega^0- T\epsilon \omega^1  +C\omega^2,\quad
  \text{where } P= D\epsilon.
\end{align}

\begin{proposition}
  \label{prop:DCH1inv}
  Up to diffeomorphism, the general solution of \eqref{eq:dom0} --
  \eqref{eq:depsilon0} depends on one function of one variable.
\end{proposition}
\begin{proof}
Writing
\begin{equation}
  \label{eq:dP}
  \d P = P_1 \omega^0 +(C-2TP)\omega^1+2(C+2T^2)\epsilon\,
  \omega^2,\quad \text{where } P_1 = D^2 \epsilon,
\end{equation}
a straightforward calculation shows that the differential ideal
generated by \eqref{eq:dom0}--\eqref{eq:depsilon0} is closed; i.e.,
$\d^2 = 0$.  The symbol tableau and its prolongation are
\[
A = \spn
  \begin{pmatrix}
    1&0&0
  \end{pmatrix},\quad
A^{(1)} =
  \spn \begin{pmatrix}
    1 & 0 &0\\
    0 & 0 &0\\
    0&0&0
  \end{pmatrix}.
\]
Hence, the reduced characters are $c_1=1, c_2=0, c_3=0$, with
\[
c_1+2c_2+3c_3= 1=\dim A^{(1)}.
\]
The tableau is involutive of rank 1. The
desired conclusion
now follows by  \cite{bryant1}.
\end{proof}
\begin{proposition}
  \label{prop:DCH1q5}
  Generically, the metric described by the preceding Proposition is
  classified by 5th order invariants.
\end{proposition}
\begin{proof}
  For generic solutions of \eqref{eq:dom0}-\eqref{eq:depsilon0},
  $\epsilon, P=D\epsilon, P_1=D^2\epsilon$ are functionally
  independent.  We already observed that $\epsilon$ is a 2nd order
  invariant.  Hence $P,P_1$ are a 3rd and a 4th order invariant,
  respectively.  Generically, these will be functionally independent,
  and therefore, the rank sequence is as shown in
  \eqref{eq:CH1maxorder}.
\end{proof}

\section{Three-dimensional metrics of maximal order}\label{sec:3Dmax}
In this section, we prove Theorem \ref{thm:main}.  Throughout, we
assume full rank regularity and curvature regularity.  By
Propositions \ref{prop:NCH1} - \ref{prop:DCH1}, all 4th order
metrics are necessarily type D and properly $\CH_1$. By Propositions
\ref{prop:DCH1inv} and \ref{prop:DCH1q5} such a geometry satisfies
\eqref{eq:dom0}-\eqref{eq:depsilon0} and
\begin{align}
  \label{eq:invindep}
  &\d\epsilon\wedge \d P \wedge \d P_1 \neq 0,\quad C+2T^2\neq
  0,\intertext{where}
  \label{eq:defPP1}
  &P:= D\epsilon,\quad
  P_1 := DP = D^2 \epsilon.
\end{align}
\noindent
We complete the proof of the main Theorem \ref{thm:main} by
integrating \eqref{eq:dom0}-\eqref{eq:depsilon0} subject to the
constraints \eqref{eq:invindep}.

First assume $T\neq 0$. To integrate \eqref{eq:dom0} we introduce an
integrating factor:
\[ \d (e^{-2Tw} \omega^0) =  0. \]
Hence,
\begin{align}
  \label{eq:om0du}
  \omega^0 &= e^{2Tw} \d u,\\
  \label{eq:om1dw}
  \omega^1 &= 2\d w +4 T x \d u,
\end{align}
for some functions $u,x,w$. Next, \eqref{eq:dom1} gives
\[ \d\omega^1 - 4T
  \omega^2 \wedge \omega^0 = 4T(e^{-2 T w}dx - \omega^2)\wedge
  \omega^0 = 0 ,\]
with general solution
\begin{align}
  \label{eq:om2dx}
  \omega^2 &= e^{-2Tw} (\d x + a \d u) .
\end{align}
Since $\omega^0, \omega^1,\omega^2$ are linearly independent, $u,w,x$
form a system of coordinates, and $a$ is some, as yet
undetermined, function of $u,w,x$.  Solving \eqref{eq:om0du}
\eqref{eq:om1dw} \eqref{eq:om2dx} gives
\begin{align*}
  \d u&= e^{-2Tw} \omega^0,\\
  \d w &= \frac{1}{2} \omega^1 - 2T x\, e^{-2 Tw}\omega^0,\\
  \d x &= e^{2Tw} \omega^2 -a\, e^{-2Tw} \omega^0.
\end{align*}
By
\eqref{eq:depsilon0}, we have
\begin{equation}
  \d(e^{2Tw} \epsilon - C x)\wedge \omega ^0 = 0.
\end{equation}
Hence,
\begin{equation}
  \label{eq:epsilonfu}
  \epsilon = e^{-2Tw}\left(Cx + f(u)\right),
\end{equation}
for some univariate function $f(u)$.  Taking the exterior derivative
of \eqref{eq:om2dx} and using \eqref{eq:dom2} gives
\begin{equation}\label{eq:da}
  \left\{ \d a+ 2 e^{4Tw} \d w + 2((C-2T^2)x + f(u))\d x\right\}\wedge \d u =
  0.
\end{equation}
Making the substitution
\begin{align}
  &a=a_1+ \frac{1-e^{4Tw}}{2T}+ x^2(2T^2-C) - 2x f(u) ,\intertext{gives}
  &\d a_1\wedge \d u = 0.
\end{align}
Therefore, for $T\neq 0$ the general solution of
\eqref{eq:dom0}-\eqref{eq:depsilon0} is given by \eqref{eq:om0du},
\eqref{eq:om1dw}, \eqref{eq:om2dx} and
\begin{equation}
  \label{eq:aff1}
  a=\frac{1-e^{4Tw}}{2T}+ x^2(2T^2-C) - 2x f(u)+f_1(u),
\end{equation}
where $f(u), f_1(u)$ are freely chosen functions.
This solution form is invariant with respect to the following
transformations:
\begin{align}
  &u = \phi(U),\quad w = W- \frac{\log \phi'}{2T},\quad  x =
  \frac{X}{\phi'} + \frac{\phi''}{4T^2(\phi')^2},\\
  \label{eq:fugeneric}
  &f(u) =  \frac{F(U)}{\phi'} - \frac{C \phi''}{4T^2(\phi')^2},\\
  &f_1(u) = \frac{F_1(U)}{(\phi')^2}+ \frac{2 F(U) \phi'' - \phi'''}{4
    T^2 (\phi')^3}+ \frac{(6T^2-C) (\phi'')^2}{16 T^4 (\phi')^4} +
  \frac{1-(\phi')^2}{2T (\phi')^2},
\end{align}
where $\phi(U)$ is an arbitrary strictly increasing function
($\phi'(U)>0$ everywhere).

If $T=0$ then one verifies that \eqref{eq:dom0}-\eqref{eq:depsilon0}
is still equivalent to \eqref{eq:om0du}-\eqref{eq:da}. Moreover, if
$(1-e^{4Tw})/(2T)$ is interpreted in the limit sense as being equal to
$-2w$, \eqref{eq:aff1} remains valid.  The form-preserving
transformations are now
\begin{align}
  &u = U + U_0,\quad w = W+W_0,\quad x = X+ \phi(U),\\
  \label{eq:fuC0}
  & f(u) = F(U) - C \phi,\quad f_1(u) = F_1(U) + 2F(U)\phi - C\phi^2 -
  \phi' + 2W_0,
\end{align}
where $U_0, W_0$ are constants and $\phi(U)$ is an arbitrary
function.



It follows by \eqref{eq:fugeneric} and \eqref{eq:fuC0} that if $C\neq
0$, then one can normalize the above solution form by transforming
$f(u) \to 0$ identically.  If $C=0$ then $T\neq 0$ by assumption, and
hence by \eqref{eq:epsilonfu} and \eqref{eq:fugeneric} one can
normalize the solution form by transforming $f(u) \to
2T^2$. Evaluating $\frac{1}{2} (\omega^1)^2 - 2 \omega^0 \omega^2 $
gives the metric in \eqref{eq:ds1}.  Finally, a straightforward
calculation relative to this metric form shows that the maximal order
condition \eqref{eq:invindep} is equivalent to
\eqref{eq:Fu-maxorder-cond}.

The above maximal order metrics are invariantly classified by the
invariant scalars $C,T$ and by the following  Cartan invariants of
orders $2,3,4,5$, respectively:
\[ \epsilon, P= D\epsilon, P_1 = D^2\epsilon, P_2 = D^3 \epsilon.\]
If $C\neq 0$, it is convenient to introduce the invariants
\begin{align*}
  \label{eq:Jdef}
  J &:= \delta \epsilon \Delta P - \delta P \Delta \epsilon=2CTP-2T(C+2T^2)\epsilon^2-C^2,\\
  J_1 &:= D \epsilon \Delta P - D P \Delta \epsilon=-CP_1+2(C+2T^2)\epsilon P,\\
  J_2 &:= D \epsilon \Delta J_1 - D J_1 \Delta
  \epsilon\\
  &=C^2P_2-2C\epsilon(C+2T^2)P_1-2C(C+6T^2)P^2+4T[C^2+2T(C+2T^2)\epsilon^2]P.\nonumber
\end{align*}
Hence, the invariants $J_1$ and $J_2$ have
order 4 and 5 respectively. If $T=0$ then $J=-C^2$ is constant. In
the generic case $CT\neq 0$, and in the light of \eqref{eq:dP} and
analogous structure equations for $dJ$, the invariant $J$ is
non-constant and of order 3.
Explicit calculations relative to the metric form \eqref{eq:ds1}
show that
\begin{align*}
  &\epsilon = C e^{-2 Tw} x,\\
  &J = - C^2 e^{-4Tw}(1+2 T F(u)),\\
  &A:=\frac{(CJ_1+4T\epsilon J)^2}{J^3}  = - \frac{(F'(u))^2}{(1+2
    T F(u))^{3}},\\
  &B:=\frac{CJ_2+20C J_1 T^2 \epsilon}{J^2} + \frac{48 T^3\epsilon^2}{J} =
  - \frac{F''(u)}{(1+2 T F(u))^2}.
\end{align*}
The latter two invariants have order 4 and 5, respectively.  The
metric is classified by the functional relationship between these
invariants.  Observe that the maximal order
condition is $B\neq 3TA$.

If $C=0\neq T$ we define dimensionless invariants of order 3, 4 and
5:
\[
p:=P/\epsilon^2,\qquad p_1:=-P_1/\epsilon^3,\qquad
p_2:=P_2/\epsilon^4.
\]
Explicit calculations relative to \eqref{eq:ds1} now give
\begin{align*}
  &\epsilon = 2 T^2 e^{-2Tw}, \qquad p=2x,\\
  &U:=p_1+\frac{3p^2}{2}+2T^2\left(p+\frac{T}{\epsilon^2}\right)=\frac{F(u)}{T^2} + \frac{1}{2T^3},\\
  &V:=p_2+2(3p+1)p_1+6p^3+4p\left(p+\frac{T}{\epsilon^2}\right)=-\frac{F'(u)}{2T^4}.
\end{align*}
Hence, as above, the metric is classified by the functional
relationship between a 4th and a 5th order invariant. The maximal
order condition is $V\neq0$.

\section{Conclusions and discussion}
\label{sec:discussion} In this article we have demonstrated that
3-dimensional Lorentzian metrics may require 5th order differential
invariants for their invariant classification. The class of maximal
order metrics consists of a single, well-defined family of $\CH_1$
solutions governed by a unified set of structure equations.  This
echoes a similar result in 4-dimensional Lorentzian
geometry~\cite{milpel08}, although there the possibility of
pseudo-stable geometries of maximal order was left open.

Previously, 3-dimensional Lorentzian $\CH_1$ metrics were studied in
detail by Bueken and Djoric~\cite{BD00}.  They
 already proved Proposition \ref{prop:DZCH1} and obtained the
metrics covered by Propositions \ref{prop:NCH1} and \ref{prop:DCH1},
albeit not in closed form but up to solving partial differential
equations. The coordinate forms in \cite{BD00} are therefore less
convenient for invariant classification and the discussion of the
order, whereas our work was more directly related to Cartan
invariants. Even though our focus here was on type D metrics of
maximal order, the type N, 3rd order $\CH_1$ geometries from
Proposition \ref{prop:NCH1} also constitute an interesting class
governed by a well-defined set of structure equations. A closed form
for these metrics can be derived along the same lines as in the type
D case outlined above, but we do not pursue this here.

In \cite{cps10} it was proved that the unique TMG solution of type D
(dubbed type D$_s$ there, cf.\ table \ref{tab:petrovtype} of
appendix C) is the homogeneous, biaxially spacelike-squashed AdS$_3$
metric family; this is the unique solution corresponding to the
proof of Proposition \ref{prop:DCH0}. Type D NMG solutions with
constant scalar curvature were fully classified in \cite{aa11a} and
are also homogeneous.
Hence, the metrics of Theorem \ref{thm:main} are not TMG nor NMG
solutions.  Therefore, our conclusion is  that
\begin{quote}
  \em at most four covariant derivatives of the Riemann tensor are
  needed to invariant classify exact TMG and NMG solutions locally.
\end{quote}
 In future work, we want to sharpen this result.
Hereby, the technique we have followed in this paper to prove
Propositions \ref{prop:NCH1}-\ref{prop:DCH1}
not only provides a robust mechanism to invariantly characterize
solutions, but also allows one to find new solutions, beyond the
curvature homogeneity assumption. A first step, however, would be to
classify all curvature homogeneous TMG and NMG solutions, in order
to see whether the bound $q-1\leq 3$ for the TMG and NMG
gravitational theories is sharp.


Finally, the same argument given for Proposition
\ref{prop:pseudostab} holds for Riemannian geometry as well.
However, for Euclidean signature only the equivalent of type DZ
curvature is possible and this suffices to rule out 4th order
Riemannian metrics.  However 3rd order, 3-dimensional order
Riemannian metrics are possible. We will report on this fact
elsewhere.

\section*{Acknowledgments}

RM was supported by an NSERC discovery grant.  He thanks the
Mathematical Institute of Utrecht University for its hospitality
during a research visit. LW was supported by a BOF research grant of
Ghent University, an FWO mobility grant No V4.356.10N to Utrecht
University and an Yggdrasil mobility grant No 211109 to University
of Stavanger while parts of this work were performed. He thanks the
Department of Mathematics and Statistics of Dalhousie University for
its hospitality during a research stay.

\appendix
\section{The three-dimensional formalism}\label{app:NPformalism}
Several three-dimensional NP-like formalisms, with different symbol
choices, have been proposed in the context of exact solutions to
topologically massive gravity \cite{hmp87,an95}.  Our choice of
symbols is close to \cite{an95}, but differs slightly in the choice
of normalization because we attempted to satisfy the following
criteria:
\begin{itemize}
\item our 3-dimensional formalism is obtainable as a straightforward reduction
  of the usual 4-dimensional NP formalism~\cite[Chapter 7]{ES};
\item our rule for passing from vector to spinor indices is very
  simple and does not involve normalizing factors;
\item the relation between the curvature spinor and the Ricci tensor
  takes a particularly simple form; see  equation \eqref{eq:SabPsi}.
\end{itemize}

Let $(U,\varepsilon)$ be a 2-dimensional symplectic, real vector
space. The group of symplectic automorphisms is isomorphic to
$\SL_2\Rset$.  The vector space $V=\S^2U$ carries the natural
structure of a Lorentzian inner product space with the inner product
given by $\eta=-\varepsilon^2$.  Henceforth, we regard $U$ as the
space of spinors and $V$ as the space of vectors.  The group $O(\eta)$
is isomorphic to $\SO(1,2)$; the group morphism $\SL_2\Rset\to
\SO(1,2)$ gives the double cover of vectors by spinors.

To facilitate frame calculations, we introduce a normalized spinor
dyad $\bo, \biota$:
\begin{eqnarray*}
&&\varepsilon_{01}\equiv \varepsilon_{AB} o^A \iota^B \equiv  o_A
\iota^A = -\iota_A o^A \equiv -\varepsilon_{10} = 1,\\
&&\varepsilon_{00}\equiv o_A o^A =0,\quad
\varepsilon_{11}\equiv\iota_A\iota^A =0,
\end{eqnarray*}
 where the dyad indices $A,B,\ldots$ take values 0 or 1.
Associated to this dyad, we define a null vector triad  by
\begin{equation}\label{eq:dyad_triad}
\be_0 = \bell = \bo^2,\quad \be_1=\bm = \bo \biota
\equiv\frac{1}{2}\left(\bo\otimes \biota + \biota \otimes
\bo\right),\quad \be_2=\bn = \biota^2,
\end{equation}
where the triad indices $a,b,c=0,1,2$ are doublets of symmetrized
dyad indices:
\[ 0 \mapsto (00),\quad  1\mapsto (01),\quad 2\mapsto (11).\]
In this way, we have
\begin{gather}
  \eta_{ab} = \eta_{(A_1 A_2) (B_1 B_2)} = -\frac{1}{2}(\varepsilon_{A_1 B_1}
  \varepsilon_{A_2
    B_2}+\varepsilon_{A_1 B_2} \varepsilon_{A_2 B_1}),\\
  \eta_{02} = \eta_{20} = -1,\quad \eta_{11} = 1/2,
\end{gather}
with all other components zero.  Equivalently, \[ \eta_{ab} \ell^a
n^b = \ell_a n^a = -1,\quad m_a m^a = 1/2,\] with all other inner
products equal to zero.

Next, let $(M,g)$ be a 3-dimensional Lorentzian manifold. A null
triad at $x\in M$ is an isomorphism $(V,\eta)\to (T_x M, g_x)$. A
moving $\eta$-frame is a null triad at every $x\in O$ for some open
neighbourhood $O\subset M$.  Equivalently, a null triad is a
collection of vector fields $\bell, \bm, \bn$  that satisfy the
relations
\[ g(\bell,\bn) = -1,\quad g(\bm,\bm) = 1/2,\] with all other
inner products zero.  In other words, taking $(\be_0,\be_1,\be_2) =
(\bell,\bm,\bn)$ gives 
\begin{equation}
\label{eq:gab}
  (g_{ab}) =(\eta_{ab})=
\begin{pmatrix}
  0 &  0 & -1\\ 0 & 1/2 & 0\\ -1 &
  0 & 0
\end{pmatrix},\qquad
 (g^{ab}) =(\eta^{ab})=
\begin{pmatrix}
  0 &  0 & -1\\ 0 & 2 & 0\\ -1 &
  0 & 0
\end{pmatrix}.
\end{equation}

In introducing symbols for the connection scalars, we wish to adapt
the notation of the familiar four-dimensional NP formalism.  To do so,
it is convenient to regard the manifold $M$ as a totally geodesic
embedding (all geodesics in the submanifold are also geodesics of the
surrounding manifold) $\phi: M \hookrightarrow \hM$ in a
4-dimensional Lorentzian manifold $(\hM,\hat{g})$. This is equivalent
to the condition that $M$ be autoparallel, i.e., that the covariant
derivative operator is closed with respect to vector fields that are
tangent to $M$ \cite[Chapter 7, Sect.  8]{kobnam}.

Recall that a null tetrad framing on $\hM$ is a basis of vector
fields $(\hat\bm,\hat\bm^*,\hat\bn,\hat\bell)$ such that
\[ \hat{g}(\hat\bm,\hat\bm^*) = 1,\quad \hat{g}(\hat\bell,\hat\bn) =
-1, \] with all other cross-products equal to zero. Here $\hat\bell,
\hat\bn$ are real whereas $\hat\bm, \hat\bm^*$ are complex
conjugates. We relate the null tetrad on $\hM$ to the null triad
on $M$ by setting
\begin{equation}
  \label{eq:tetradpushforward}
  \phi_*\bell = \hat{\bell},\quad \phi_*\bm = \Re\hat{\bm} =
  (\hat\bm+\hat\bm^*)/2,\quad \phi_*\bn = \hat{\bn}.
\end{equation}
Let $\hat{\bomega}^i,\; i=1,2,3,4$ and $\hat{\bGamma}_{ij}$ denote the
dual coframe and the connection 1-form on $\hM$. Let
\[ \tbomega^i = \phi^* \hat{\bomega}^i,\quad \tbGamma_{ij} =
\phi^*\hat{\bGamma}_{ij} \] denote the corresponding pullbacks to
$M$.  Henceforth, we use a tilde decoration to denote the pullback
of objects from $\hM$ to $M$. The pullback imposes the
condition:
\begin{equation}
  \label{eq:tbomega12}
  \tbomega^1=\tbomega^2.
\end{equation}
The embedding of $M$ into $\hM$ induces an inclusion of the three-dimensional
Lorentz group $\SO(1,2)$ into $\SO(1,3)$, the four-dimensional Lorentz group.  The
condition that $M$ be autoparallel is equivalent to the condition
that the pull-back of the connection 1-form take values in the
subalgebra $\so(1,2)$.  This imposes the following conditions on the
pullback of the connection 1-form:
\[ \tbGamma_{14} = \tbGamma_{24},\quad \tbGamma_{23} =
\tbGamma_{13},\quad \tbGamma_{12}=0 .\] Using the notation of
\cite[Section 7.2]{ES}, the corresponding condition on the NP
connection scalars is:
\begin{gather}
  \label{eq:imspincoeffzero1}
  \Im \tilde{\kappa} = \Im\tilde{\tau} =
  \Im\tilde{\epsilon}=\Im\tilde\gamma = \Im\tilde{\pi} =\Im\tilde\nu =
  0,\\
    \label{eq:imspincoeffzero2}
    \Im(\tilde\sigma+\tilde\rho) =\Im(\tilde\alpha+\tilde\beta) =
  \Im(\tilde\lambda + \tilde\mu) = 0.
\end{gather}
Taking into account the
difference in the ordering of the three-dimensional and the four-dimensional indices, we arrive at
the following notation for the three-dimensional connection 1-form and scalars:
\begin{gather}
  \bomega^0 = \tbomega^4,\quad \bomega^1 = 2\tbomega^1 =
  2\tbomega^2,\quad \bomega^2 = \tbomega^3;\\
  \label{eq:Gammamatrix}
  \left(\bGamma^a{}_{b}\right) =
  \begin{pmatrix}
   \bGamma_{02} & \bGamma_{12} & 0 \\
   -2 \bGamma_{01}  & 0 & 2 \bGamma_{12} \\
   0 & -\bGamma_{01} &  -\bGamma_{02}
 \end{pmatrix},
\end{gather}

\begin{align}
  \bGamma_{01} = -\tbGamma_{14} &= \kappa \bomega^0 + \sigma\bomega^1
  + \tau\bomega^2,\\
  \bGamma_{02} = -\tbGamma_{34} &= 2\left(\epsilon\bomega^0 +
    \alpha\bomega^1+
    \gamma\bomega^2\right) ,\\
  \bGamma_{12} = \hphantom{-}\tbGamma_{23} &= \pi\bomega^0 +
  \lambda\bomega^1+\nu\bomega^2;
\end{align}
\begin{align}
  \label{eq:kappatkappa}
  & \kappa = \tilde\kappa ,&& \tau = \tilde\tau ,&& \sigma =
  (\tilde\sigma+\tilde\rho)/2 ,\\
  & \pi = \tilde\pi,&& \nu = \tilde\nu ,&& \lambda =
  (\tilde\lambda+\tilde\mu)/2 ,\\
  \label{eq:epsilontepsilon}
  & \epsilon = \tilde\epsilon,&& \gamma =
  \tilde\gamma ,&& \alpha =
  (\tilde\alpha+\tilde\beta)/2.
\end{align}
Writing
\begin{equation}
  \label{eq:DdeltaDeltadef} D = \ell^a \nabla_a,\quad \delta =
  m^a\nabla_a,\quad \Delta = n^a\nabla_a,
\end{equation}
we have by \eqref{eq:tetradpushforward}:
\begin{equation}
  \label{eq:3d4dpushforward}
  D\tilde\psi = \phi^*(\hat{D}\hat\psi),\quad
  \delta\tilde\psi = \phi^*(\hat\delta\hat\psi+\hat\delta^*\hat\psi)/2,\quad
  \Delta\tilde\psi = \phi^*\hat\Delta\hat\psi,
\end{equation}
where $\hat\psi$ is a scalar defined on $\hM$ and
$\tilde\psi=\phi^*\hat\psi$ is its pullback to $M$.  The three-dimensional
commutator relations can now be expressed as
\begin{align}
  \label{eq:commutator1}
  D\delta - \delta D &= (\pi-2\alpha)D +2\sigma \delta-\kappa\Delta,\\
  \label{eq:commutator2}
  D\Delta - \Delta D &= -2\gamma D +2 (\tau+\pi) \delta -2\epsilon \Delta,\\
  \label{eq:commutator3}
  \delta \Delta-\Delta \delta &= -\nu D +2\lambda \delta +(\tau-2\alpha)
  \Delta.
\end{align}
The above equations follow in a straightforward manner by applying
symbol rules \eqref{eq:kappatkappa}-\eqref{eq:epsilontepsilon},
\eqref{eq:3d4dpushforward} to the usual four-dimensional commutator relations, as
shown for example in equations (7.6a)-(7.6c) of \cite{ES}.

The three-dimensional curvature tensor $R_{abcd}$ decomposes into a curvature
scalar
\begin{equation}\label{eq:defRab}
R \equiv R^a{}_a,\quad R_{ab}\equiv R^c{}_{acb},
\end{equation}
and a trace-free part
\begin{equation}\label{eq:defSab}
S_{ab} \equiv R_{ab}- \frac{1}{3}R g_{ab} ,
\end{equation}
 according to
\[ R_{abcd} =(S_{ac} g_{bd}+ S_{bd} g_{ac}-S_{bc} g_{ad} -
S_{ad}g_{bc}) + \frac{1}{6} R (g_{ac} g_{bd} - g_{ad} g_{bc}). \] The
image of the natural inclusion $\S^4 U \hookrightarrow \S^2 V$ is the
5-dimensional vector space of trace-free, symmetric tensors.
Therefore, the trace-free part of a three-dimensional curvature tensor can be
represented by means of a rank-4, symmetric curvature spinor:
\[ \Psi_{ABCD} = (\Psi_0 \biota^4 + 4\Psi_1 \biota^3 \bo + 6 \Psi_2
\biota^2 \bo^2 + 4 \Psi_3 \biota \bo^3 + \Psi_4 \bo^4)_{(ABCD)}. \]
In this way, the definition of the curvature scalars $\Psi_0,
\Psi_1, \Psi_2, \Psi_3, \Psi_4$ is formally identical to their
four-dimensional counterparts; c.f.\ \cite[Equation (3.76)]{ES}.  We
obtain the following representation of the trace-free part of the
Ricci tensor and the curvature-two form:
\begin{gather}
  \label{eq:SabPsi}
  (S_{ab}) =
\begin{pmatrix}
  \Psi_0 &  \Psi_1 & \Psi_2\\ \Psi_1 & \Psi_2 & \Psi_3\\ \Psi_2 &
  \Psi_3 & \Psi_4
\end{pmatrix};\\
 \left(\bOmega^a{}_{b}\right) =
 \begin{pmatrix}
   \bOmega_{02} & \bOmega_{12} & 0 \\
   -2 \bOmega_{01}  & 0 & 2 \bOmega_{12} \\
   0 & -\bOmega_{01} &  -\bOmega_{02}
 \end{pmatrix},
\end{gather}
\begin{align}
  \label{eq:O01}
  \bOmega_{01} &= \frac{1}{2} \Psi_0\, \bomega^0\wedge \bomega^1 +
  \Psi_1\, \bomega^0\wedge \bomega^2 +\left(\Psi_2/2+R/12\right)
  \bomega^1\wedge \bomega^2,\\
  \label{eq:O02}
  \bOmega_{02} &= \Psi_1\, \bomega^0\wedge \bomega^1 +
  \left(2\Psi_2-R/6\right)\,
  \bomega^0\wedge \bomega^2 +\Psi_3 \bomega^1\wedge \bomega^2,\\
  \label{eq:O03}
  \bOmega_{12} &= \left(\Psi_2/2+R/12\right) \bomega^0\wedge \bomega^1
  + \Psi_3\, \bomega^0\wedge \bomega^2 +\frac{1}{2}\Psi_4\,
  \bomega^1\wedge \bomega^2.
\end{align}

The three-dimensional curvature 2-form and curvature scalars are related to their four-dimensional
counterparts as follows:
\begin{align}
  &\bOmega_{01} = -\tilde{\bOmega}_{14},\quad \bOmega_{02} =
  -\tilde{\bOmega}_{34},\quad \bOmega_{12} =
  \tilde{\bOmega}_{23},\quad \tilde{\bOmega}_{12} = 0;\\
  \label{eq:Psi0tPsi0}
  &\Psi_0 = \tPsi_0 + \tPhi_{00},\\
  &\Psi_1 = \tPsi_1 + \tPhi_{01},\\
  &\Psi_2 = \tPsi_2 + \tPhi_{02}/3+2/3\, \tPhi_{11},\\
  &\Psi_3 = \tPsi_3 + \tPhi_{12},\\
  &\Psi_4 = \tPsi_4 + \tPhi_{22},\\
  \label{eq:RtR}
  &R = \tilde{R}/2  +4\tPhi_{02}-4\tPhi_{11},
\end{align}
with the right-hand sides of the above equations all real, as a
consequence of equations \eqref{eq:imspincoeffzero1} and
\eqref{eq:imspincoeffzero2}.

The three-dimensional version of the NP equations, or equivalently Cartan's second
structure equations, take the form shown below.  Using equations
\eqref{eq:kappatkappa}-\eqref{eq:epsilontepsilon},
\eqref{eq:3d4dpushforward}, \eqref{eq:Psi0tPsi0}-\eqref{eq:RtR},  it
is straightforward to convert the four-dimensional Newman-Penrose equations into
their three-dimensional counterparts. For example, the NP equations (7.21a) and
(7.21b) of \cite{ES} read
\begin{align*}
  &D\rho-\delta^*\kappa =
  \rho^2+\sigma\sigma^*+(\epsilon+\epsilon^*)\rho - \kappa^*\tau -
  \kappa(3\alpha+\beta^*-\pi) +\Phi_{00},\\
  &D\sigma-\delta\kappa =
  (\rho+\rho^*)\sigma+(3\epsilon-\epsilon^*)\sigma -
  (\tau-\pi^*+\alpha^*+3\beta)\kappa +\Psi_{0}.
\end{align*}
Note that all of the symbols in the above two equations should have
hats, but we omit the decoration for the sake of simplicity.  Taking
the average of these two equations, pulling back and using
\eqref{eq:kappatkappa}-\eqref{eq:epsilontepsilon},
\eqref{eq:3d4dpushforward}, \eqref{eq:Psi0tPsi0}-\eqref{eq:RtR}
gives equation \eqref{eq:NP1} below. The rest of the
three-dimensional structure equations are obtained via the same
reduction procedure.
\begin{align}
  \label{eq:NP1}
  D\sigma -\delta\kappa&=(\pi-4\alpha-\tau)\kappa +
  2(\epsilon+\sigma)\sigma + \Psi_0/2,\\
  \label{eq:NP2}
  D\tau -\Delta\kappa&=-4\gamma \kappa+2( \tau+ \pi)
  \sigma +\Psi_1,\\
  \label{eq:NP3}
  D\alpha-\delta\epsilon &=2(\sigma-\epsilon)\alpha +
  (\epsilon+\sigma)\pi - (\gamma+ \lambda) \kappa
  +\Psi_1/2,\\
  \label{eq:NP4}
  \delta\tau-\Delta\sigma&=2(\lambda-\gamma) \sigma -\kappa
  \nu +\tau ^2+\Psi_2/2+R/12,\\
  \label{eq:NP5}
  D\gamma-\Delta\epsilon &=2 ( \tau+ \pi) \alpha +\pi \tau -4 \gamma
  \epsilon -\kappa \nu +\Psi_2-R/12,\\
  \label{eq:NP6}
  D\lambda-\delta\pi&=2(\sigma-\epsilon) \lambda -\kappa \nu
  +\pi^2+\Psi_2/2+R/12,\\
  \label{eq:NP7}
  \delta\gamma-\Delta\alpha &= 2(\lambda-\gamma)\alpha +
  (\lambda+\gamma)\tau -(\epsilon+\sigma)\nu
  + \Psi_3/2,\\
  \label{eq:NP8}
  D\nu-\Delta\pi &= -4\epsilon\,\nu + 2(\tau+\pi)\lambda +\Psi_3,\\
  \label{eq:NP9}
  \delta\nu-\Delta\lambda &= (\tau-4\alpha-\pi)\nu +
  2(\gamma+\lambda)\lambda + \Psi_4/2.
\end{align}
Likewise, the differential Bianchi equations are obtained by
averaging the four-dimensional Bianchi equations and applying
equations \eqref{eq:kappatkappa}-\eqref{eq:epsilontepsilon},
\eqref{eq:3d4dpushforward}, \eqref{eq:Psi0tPsi0}-\eqref{eq:RtR}.
They are:
\begin{align}
  \label{eq:bianchi1}
  &\Delta\Psi_0/2-\delta\Psi_1 + D\left(\Psi_2/2+R/12\right) =\\  \nonumber
  &\qquad   (2\gamma-\lambda)
  \Psi_0+(\pi-2\alpha-2\tau) \Psi_1 + 3\sigma \Psi_2 - \kappa \Psi_3,\\
  \label{eq:bianchi2}
  &\Delta \Psi_1/2- \delta(\Psi_2-R/12)+D\Psi_3/2=\\ \nonumber
  &\qquad \nu \Psi_0/2+(\gamma-2\lambda) \Psi_1+(3/2)(\pi-\tau)
  \Psi_2+(2\sigma-\epsilon)\Psi_3 - \kappa \Psi_4/2,\\
  \label{eq:bianchi3}
  &\Delta(\Psi_2/2+ R/12)-\delta\Psi_3+D\Psi_4/2=\\
  &\qquad \nu \Psi_1-3\lambda \Psi_2 + (2\pi+2\alpha-\tau)\Psi_3 +
  (\sigma-2\epsilon)\Psi_4.\nonumber
\end{align}

\section{Lorentz Transformations}\label{app:Lorentz-trafos}
There are 3 different types of three-dimensional Lorentz
transformations: boosts, spins, and null rotations.  Each such
transformation has a simple description as a transformation of
spinor space, i.e., as an element of $\SL_2\Rset$.  Infinitesimally,
boosts have non-zero, real eigenvalues, spins have imaginary
eigenvalues, and null rotations have zero eigenvalues (in other
words, an infinitesimal null rotation is a nilpotent transformation
of spinor space).  To facilitate calculations, we represent these
transformations in a natural spinor dyad, and present their induced
action on a suitable associated vector triad and on the
corresponding connection and curvature scalars. Consistent with our
philosophy of concordance between the three-dimensional and
four-dimensional formalisms, all of the above equations are
straightforward reductions of the four-dimensional transformation
laws; c.f.\ \cite[Appendix B]{stewart}.

A boost transformation corresponds to a real-diagonalizable element
of $\SL_2\Rset$.  The corresponding spinor and vector actions are
\begin{gather}
  \label{eq:boost0}
  \bo' = a^{1/2} \bo,\quad \biota' = a^{-1/2} \biota,\quad a
  >0,\\ \label{eq:boost_vector} \bell' = a \bell,\quad \bm' =
  \bm,\quad \bn' = a^{-1} \bn.
\end{gather}
Boost transformations can also be realized as the 1-dimensional
group of symmetries of the type D curvature spinor; c.f.\ line 6 of
Table \ref{tab:petrovtype}. The associated connection and curvature
transformation laws are shown below.
\begin{align}
  \label{eq:boost1}
  &\tau' = \tau,\\
  &\pi' = \pi,\\
  &\sigma' = a \sigma,\\
  &\lambda' = a^{-1} \lambda,\\
  \label{eq:boost5}
  &\kappa' = a^2\kappa,\\
  \label{eq:boost6}
  &\nu' = a^{-2} \nu,\\
  &\epsilon' = a\epsilon  + Da/2,\\
  &\alpha' = \alpha + a^{-1}\delta a/2,\\
  &\gamma' = a^{-1} \gamma + a^{-2} \Delta a/2,
\end{align}
\begin{align}
  &\Psi_0' = a^2\Psi_0,\\
  &\Psi_1' = a \Psi_1,\\
  &\Psi_2' = \Psi_2,\\
  &\Psi_3' = a^{-1} \Psi_3,\\
  \label{eq:boost14}
  &\Psi_4' = a^{-2} \Psi_4.
\end{align}

A null rotation
corresponds to a unipotent, non-diagonalizable element of
$\SL_2\Rset$.  The corresponding spinor and vector actions are
\begin{gather}
  \label{eq:nullrot}
  \bo' = \bo,\quad \biota' = \biota + x\bo,\\
  \label{eq:nullrot_vector} \bell' = \bell
  ,\quad \bm' = \bm + x \bell,\quad \bn' = \bn + 2x \bm + x^2 \bell.
\end{gather}
Null rotations can also be realized as the 1-dimensional group of
symmetries of the type N curvature spinor; c.f.\ line 9 of Table
\ref{tab:petrovtype}. The associated transformation laws for the
connection and curvature scalars are shown below.
\begin{align}
  \label{eq:nr1}
  &\kappa'= \kappa ,\\
  \label{eq:nr2}
  &\sigma'= \sigma +x \kappa ,\\
  \label{eq:nr3}
  &\epsilon'= \epsilon +x \kappa ,\\
  \label{eq:nr4}
  &\tau'=\tau +2 x \sigma+x^2 \kappa ,\\
  \label{eq:nr5}
  &\alpha'=  \alpha +x (\epsilon + \sigma)+x^2 \kappa,\\
  \label{eq:nr6}
  &\gamma'= \gamma +x(2\alpha +\tau)+x^2 (\epsilon +2\sigma)+x^3 \kappa ,\\
  \label{eq:nr7}
  &\pi'= \pi+Dx +2 x \epsilon+x^2 \kappa ,\\
  &\lambda'= \lambda +\delta x+x(2\alpha+\pi+
  Dx)+x^2(2\epsilon+\sigma)+x^3\kappa,\\
  &\nu'= \nu+\Delta x+2x(\gamma+\lambda+\delta x)+\\ \nonumber &\qquad
  +x^2(4\alpha +\tau
  +\pi+Dx)+2 x^3 (\epsilon+\sigma)+x^4\kappa  ,
\end{align}
\begin{align}
  &\Psi_0'= \Psi_0,\\
  &\Psi_1'= \Psi_1+x \Psi_0,\\
  &\Psi_2'=\Psi_2+2 x \Psi_1+x^2 \Psi_0, \\
  &\Psi_3'= \Psi_3+3 x \Psi_2+3 x^2 \Psi_1+x^3 \Psi_0,\\
  \label{eq:nr15}
  &\Psi_4'=\Psi_4+4 x \Psi_3+6 x^2 \Psi_2+4 x^3 \Psi_1+ x^4 \Psi_0.
\end{align}

A spin transformation corresponds to an element of $\SL_2\Rset$ with
imaginary eigenvalues.
As such, we have
\begin{align}
  \label{eq:spin}
  &\bo'\pm i \biota' = e^{\mp it/2} (\bo \pm i \biota),\\ \nonumber
  &(\bell+\bn)' = \LHS,\\
  &(\bell-\bn\pm 2i\bm)' = e^{\mp i t} \LHS.
\end{align}
Spin transformations can also be realized as the 1-parameter group
of symmetries of the type DZ curvature spinor; cf line 7 of Table
\ref{tab:petrovtype}. The associated connection and curvature
transformation laws are shown below.

\begin{align}
  \label{eq:spin1}
  &(\gamma+\sigma-\epsilon-\lambda)' = \LHS,\\
  &(4\alpha+\kappa-\pi+\nu-\tau)' = \LHS, \\
  &(2(\gamma+\epsilon)\pm i(\kappa-\pi+\tau-\nu))' = e^{\pm it}
  \LHS,
  \\
  \label{eq:spin4}
  &(4\alpha+\pi-\kappa+\tau-\nu+\pm 2i(\epsilon-\gamma+\sigma-\lambda)) =
  e^{\pm 2it} \LHS,\\
  &(\lambda+\sigma-\gamma-\epsilon+\pm i (\pi-\tau))' = e^{\pm
    it}(\LHS -\delta t\mp(i/2)(Dt-\Delta t)),\\
  &(\kappa+\pi+\nu+\tau)'=\LHS-(Dt+\Delta t),\\
  &(\Psi_0 + 2\Psi_2+\Psi_4)' = \LHS,\\
  &(\Psi_0 -\Psi_4 \pm 2i(\Psi_1+\Psi_3))' = e^{\mp it}\LHS,\\
  \label{eq:spin8}
  &(\Psi_0-6\Psi_2 + \Psi_4\pm 4i(\Psi_1-\Psi_3))' = e^{\mp 2it}
  \LHS.
\end{align}

Finally, there are a number of discrete Lorentz transformations that
lie outside the connected component of the identity in $O(\eta)$.
Given a null frame $(\bell,\bm,\bn)$ we define
\begin{equation}\label{eq:ONdyad}
{\bf t}\equiv \tfrac{1}{\sqrt{2}}(\bell+\bn), \quad {\bf x}\equiv
\tfrac 12(\bell-\bn).
\end{equation}
The transformation laws of the connection and curvature scalars
under reflection of the vectors of the orthonormal triad $({\bf
t},{\bf m},{\bf x})$ are also relevant for our purposes and are
given below.

\noindent
Reflection of ${\bf t}$ (`time reversal'):
\begin{eqnarray}
\label{eq:reflt_def}&&{\bf t}\mapsto -{\bf t}\;\;\Leftrightarrow \;\;\bell\mapsto-\bell,\;\bn\mapsto-\bn:\\
\label{eq:reflt_conn}&&\kappa,\tau,\alpha,\pi,\nu\;\;\text{invariant},\quad \sigma,\epsilon,\gamma,\lambda\;\;\text{change sign},\\
\label{eq:reflt_curv}&&\Psi_0,\Psi_2,\Psi_4
\;\;\text{invariant},\quad \Psi_1,\Psi_3\;\;\text{change sign}.
\end{eqnarray}
\noindent Reflection of ${\bf m}$:
\begin{eqnarray}
\label{eq:reflm_def}&&{\bf m}\mapsto -{\bf m}:\\
\label{eq:reflm_conn}&&\kappa,\tau,\alpha,\pi,\nu\;\;\text{change sign},\quad \sigma,\epsilon,\gamma,\lambda\;\;\text{invariant},\\
\label{eq:reflm_curv}&&\Psi_0,\Psi_2,\Psi_4
\;\;\text{invariant},\quad \Psi_1,\Psi_3\;\;\text{change sign}.
\end{eqnarray}
\noindent Reflection of ${\bf x}$:
\begin{eqnarray}
\label{eq:reflx_def}&&{\bf x}\mapsto -{\bf x}
\;\;\Leftrightarrow\;\; \bell
\leftrightarrow \bn:\\
\label{eq:reflx_conn}&&\kappa\leftrightarrow-\nu,\quad
\sigma\leftrightarrow-\lambda,\quad\tau\leftrightarrow-\pi,\quad
\epsilon\leftrightarrow-\gamma,\quad
\alpha'=-\alpha,\\
\label{eq:reflx_curv}&&\Psi_0\leftrightarrow
\Psi_4,\quad\Psi_1\leftrightarrow \Psi_3,\quad \Psi_2'=\Psi_2.
\end{eqnarray}

\section{Petrov-Penrose classification of the three-dimensional Ricci
tensor}\label{app:Petrov-Penrose}


Let $\bell,\bm,\bn$ a null vector triad for which $\Psi_4\neq 0$. We
introduce the three-dimensional analogue of the Petrov-Penrose
classification in terms of the root configurations of the real quartic
\begin{equation}\label{eq:Psi0z}
\Psi_0(z) = \Psi_0 + 4 \Psi_1 z + 6 \Psi_2 z^2 + 4 \Psi_3 z^3 +
\Psi_4 z^4.
\end{equation}
We note that this classification forms a special case of the general
null alignment classification for tensors in arbitrary
dimensions~\cite{mcpp05}, applied here to the three-dimensional
trace-free Ricci tensor $S_{ab}$.
Hence, in addition to the analogues of Petrov types I, II, D,
III, N
(where there are 4 real solutions $z$) and type O,
we have
to account for the possibility that some or all of the roots of
$\Psi(z)$ are complex. We will denote these additional root
configurations as Petrov types IZ (two different real roots, two
complex roots), IZZ (4 complex roots), IIZ (double real root, two
complex roots), and DZ (the double roots are complex conjugate).

Table \ref{tab:petrovtype} summarizes the three-dimensional Petrov
types, corresponding Segre types of the trace-free Ricci operator
$S^a{}_b$, the notation introduced in \cite{cps10} for the latter,
and possible normalized forms; the last column shows the dimension
$s_0$ of the corresponding automorphism group. An alternate
normalized form for type IZZ is given by
\[ \Psi_1=\Psi_3=0, \; \Psi_0=\Psi_4,\; 3|\Psi_2/\Psi_0|<1, \] but it
is related to the form in the table by a Lorentz transformation.
Analogously, a Lorentz-equivalent type D canonical form is
\[
\Psi_1=\Psi_3=0, \; \Psi_0=\Psi_4=-3\Psi_2.
\]
Note that the Ricci-Petrov classification based on null alignment
refines the Ricci-Segre type classification \cite{hmp87}.  The
distinction between Petrov types I and IZZ is the order of the
timelike eigenvalue, relative to the spacelike eigenvalues. Regarding
Segre type $\{21\}$, the spacelike or timelike character of the vector
$S_{ab}l^b$, where the null vector $\bell$ lies in the 2-dimensional
generalized eigenspace but is not an eigenvector, distinguishes
between Petrov types II and IIZ.
Also note that
Petrov type O describes a constant curvature space.


\begin{table}
  \centering
  \begin{tabular}[t]{ccccc}
    Petrov Type & Segre Type&\cite{cps10}& Normalization & $s_0$\\
    I & $\{11,1\}$ &I$_{\Rset}$& $\Psi_1=\Psi_3=0,\; \Psi_0=\Psi_4,\; 3\Psi_2/\Psi_0<-1$ & 0\\
    IZ & $\{1z\bar{z}\}$ &I$_{\Cset}$& $\Psi_1=\Psi_3=0,\; \Psi_0=-\Psi_4\neq 0$, & 0\\
    IZZ & $\{11,1\}$ &I$_{\Rset}$&$\Psi_1=\Psi_3=0, \; \Psi_0=\Psi_4,\;
    3\Psi_2/\Psi_0>1 $ & 0\\
    II  &$ \{21\}$ & II &$\Psi_1=\Psi_3=\Psi_4=0,\; \Psi_2/\Psi_0<0$ & 0\\
    IIZ& $\{21\}$ &II&$\Psi_1=\Psi_3=\Psi_4=0,\; \Psi_2/\Psi_0>0$ & 0\\
    D & $\{(11),1\}$ &D$_s$& $\Psi_0=\Psi_1=\Psi_3=\Psi_4=0,\;\Psi_2\neq 0$ & 1\\
    DZ & $\{1(1,1)\}$ &D$_t$& $\Psi_1=\Psi_3=0, \; \Psi_0=\Psi_4=3\Psi_2\neq 0$& 1\\
    III & $\{3\}$ &III&$\Psi_0=\Psi_1=\Psi_2=\Psi_4=0,\; \Psi_3=1$ & 0\\
    N & $\{(21)\}$ &N& $\Psi_0=\Psi_1=\Psi_2=\Psi_3=0,\; \Psi_4=\pm 1$ &  1\\
    O & $\{(11,1)\}$&O& $\Psi_0=\Psi_1=\Psi_2=\Psi_3=\Psi_4=0$ & 3\\[4pt]
  \end{tabular}
  \caption{The three-dimensional Petrov-Segre type}
  \label{tab:petrovtype}
\end{table}

\section{$\CH_1$ structure equations}
\label{sec:ch1}
This appendix is devoted to an analysis of the algebraic data and the
structure equations that underly curvature homogeneous geometries. In
what follows a crucial, albeit technical, innovation allows us to
simplify the form of higher order Cartan invariants by replacing them
with certain connection scalars. The general theory is detailed in
\cite{milpel09}.  For the sake of concreteness we limit the discussion
to the case of $\CH_1$ geometries.
We begin by recalling some preliminary notation
and theory, and then turn to the description of $\CH_1$ data and
structure equations, which we call a $\CH_1$ configuration.

Let $\be_a,\; a=1,\ldots, n$ be a basis of $V$, and $\bA_\alpha,\;
\alpha=1,\ldots, n(n-1)/2$ a basis of $\fo(\eta)$.  Let
$A^a{}_{b\alpha}$ denote the matrix components of $\bA_\alpha$; i.e.,
\[ \bA_\alpha \cdot \be_b = A^a{}_{b\alpha} \be_a.\]
Let $C^\alpha{}_{\beta\gamma}$ be the corresponding structure
constants:
\[  [\bA_\beta, \bA_\gamma ] = C^\alpha{}_{\beta\gamma} \bA_\alpha,\quad
A^a{}_{e\beta} A^e{}_{b\gamma} - A^a{}_{e\gamma} A^e{}_{b\beta} =
A^a{}_{b\alpha} C^\alpha{}_{\beta\gamma}.
\]

Let $\omega^a$ be an $\eta$-orthogonal coframe, $\Gamma^\alpha,
\Omega^\alpha$ the corresponding connection 1-form and curvature
2-form, respectively.  The latter are determined by the first and second
structure equations:
\begin{align}
  \label{eq:1stseq}
  \d \omega^a &= - A^a{}_{b\alpha} \Gamma^\alpha\wedge\omega^b,\\
  \label{eq:2ndseq}
  \d \Gamma^\alpha &=-\frac{1}{2} C^\alpha{}_{\beta\gamma} \Gamma^\beta\wedge
  \Gamma^\gamma  +  \Omega^\alpha,
\end{align}
where $\Gamma^\alpha{}_a,R_{abcd}$ are the connection and curvature
components, respectively:
\begin{align}
  &\Gamma^\alpha = \Gamma^\alpha{}_a \omega^a,\\
  &\Omega^\alpha = \frac{1}{2} R^\alpha{}_{cd} \, \omega^c\wedge
  \omega^d,\quad R_{abcd} = A_{ab\alpha}R^\alpha{}_{cd}.
\end{align}
The exterior derivative gives the algebraic and differential Bianchi
relations:
\begin{align}
  \label{eq:algbianchi}
  &A^a{}_{b\alpha}\,\Omega^\alpha\wedge \omega^b= 0,\\
  \label{eq:diffbianchi}
  &d\Omega^\alpha =  \frac{1}{2} C^\alpha{}_{\beta\gamma}
  \Omega^{\beta}\wedge\Gamma^{\gamma}.
\end{align}
In Appendix \ref{app:NPformalism} we introduced a convenient
formalism that assigns specific symbols to the $\Gamma^\alpha{}_a,
R_{abcd}$ when $n=3$. Our three-dimensional formalism is a suitable
reduction of the well-known 4-dimensional Newman-Penrose (NP)
formalism.  In this reduced, 3-dimensional formalism, the first
structure equations \eqref{eq:1stseq} correspond to the commutator
relations \eqref{eq:commutator1}-\eqref{eq:commutator3}; the second
structure equations correspond to reduced NP equations
\eqref{eq:NP1}-\eqref{eq:NP9}.  The component versions of the
Bianchi relations are given by
\eqref{eq:bianchi1}-\eqref{eq:bianchi3}.  The details of the
formalism and of the reduction from 4 to 3 dimensions were given in
Appendix \ref{app:NPformalism}.

Next, suppose that the $\CH_1$ condition holds and let $\omega^a$ be
a curvature normalized $\eta$-orthogonal coframe.  By the $t_1=0$
assumption,
\begin{equation}
  \label{eq:R0=tR0}
  R_{abcd} = \tR_{abcd},\quad R_{abcd;e} = \tR_{abcde},
\end{equation}
where the right hand sides denote arrays of constants.  Let
$G_0\subset O(\eta)$ be the automorphisms of $\tR_{abcd}$ and $G_1
\subset G_0$ the automorphisms of $\tR_{abcde}$.  Hence,
\eqref{eq:R0=tR0} fixes the choice of coframe up to a $G_1$ gauge
transformation.  Set
\[ \fg_{-1}:= \fo(\eta),\quad s_{-1} = \dim \fg_{-1} = n(n-1)/2.\]
Let $\fg_0,\fg_1$ denote the Lie algebra of $G_0, G_1$ respectively.
Introduce an adapted basis of $\fg_1\subset\fg_0 \subset \fg_{-1}$
consisting of
\[ (\bA_\xi, \bA_\lambda, \bA_\rho),\quad \xi = 1,\ldots s_1,\;
\lambda = s_1+1,\ldots s_0,\; \rho = s_0+1,\ldots, s_{-1},\] where
the $\bA_\xi$ are a basis of $\fg_1$, the $\bA_\lambda$ are a basis
of $\fg_{0}/\fg_{1}$ and the $\bA_\rho$ are a basis of $\fg_{-1}/
\fg_0$. By the usual definition of the covariant derivative one has
\begin{equation}\label{def:covder-Riemann} R_{abcd;e}  =  R_{abcd,e} +\Gamma^\alpha{}_e (\bA_\alpha \cdot
R)_{abcd},
\end{equation}
where for a rank $k$ tensor $T_{a_1\ldots a_k}$ the notation
\[ (\bA\cdot T)_{a_1\ldots a_k} = - \sum_{i=1}^k A^b{}_{a_i} T_{a_1
  \cdots \widehat{a_i} b \cdots a_k} ,\quad \bA\in \End(V)\] denotes
the infinitesimal action of a linear transformation on a covariant
tensor.

Since the $R_{abcd}$ are constant and since $\fg_0$ is the
annihilator of $\tR_{abcd}$, \eqref{def:covder-Riemann} becomes
\begin{equation}
  \label{eq:R1Gamma}
  \tR_{abcde} = \Gamma^\rho{}_e (\bA_\rho \cdot \tR)_{abcd}.
\end{equation}
Equation \eqref{eq:R1Gamma} describes a linear system in
$\Gamma^\rho{}_e$ with maximal rank.  Hence,
$\Gamma^\rho{}_e=\tGamma^\rho{}_e$, where the latter are constants
rationally dependent on $\tR_{abcd}, \tR_{abcde}$. Therefore, a
$\CH_1$ geometry is determined by constants
$\tR^\alpha{}_{ab}=-\tR^\alpha{}_{ba}$ and constants
$\tGamma^\rho{}_a$  such that the following relations hold, relative
to a normalized $\eta$-orthogonal coframe:
\begin{align}
  \label{eq:tRabcd}
  \tR_{abcd} &= A_{ab\alpha} \tR^\alpha{}_{cd}, \\
  \label{eq:tRabcde}
  \tR_{abcde} &= \tGamma^\rho{}_a (\bA_\rho \cdot \tR)_{abcd}.
\end{align}
Moreover, the Bianchi relations \eqref{eq:algbianchi},
\eqref{eq:diffbianchi} impose the following linear, respectively,
bilinear constraints on the above constants:
\begin{align}
  \label{eq:configabianchi}
  &\tR^\alpha{}_{[bc}  A^a{}_{d]\alpha} = 0,  \\
  \label{eq:configdbianchi}
  &(\bA_\rho\cdot \tR)^\alpha{}_{[ab} \tGamma^\rho{}_{c]} = 0.
\end{align}

Similar to \eqref{eq:R1Gamma}, the second order derivative of
curvature is given by
\begin{equation}
  \label{eq:R2Gamma}
  R_{abcd;ef}  = \tGamma^\rho{}_f (\bA_\rho\cdot \tR)_{abcde} +
  \Gamma^\lambda{}_f (\bA_\lambda \cdot \tR)_{abcde},
\end{equation}
relative to a normalized coframe.  Since the residual frame freedom is
$G_1$, the scalars $\Gamma^\lambda{}_a$ obey an algebraic,
$G_1$-transformation law. The second structure equations
\eqref{eq:2ndseq} impose the following linear constraints on these
scalars:
\begin{equation}
  \label{eq:confignp}
  \tR^\rho{}_{ab} =  C^\rho{}_{{\rho_1}{\rho_2}} \tGamma^{\rho_1}{}_a
  \tGamma^{\rho_2}{}_b-2  \tGamma^\rho{}_c \tGamma^{\rho_1}{}_{[a} A^c{}_{b]{\rho_1}} -
  2(\bA_\lambda\cdot \tGamma)^\rho{}_{[a} \Gamma^\lambda{}_{b]},
\end{equation}
where $\rho_1,\rho_2 = s_0+1,\ldots, s_{-1}$ have the same range
as $\rho$, and the following differential relations:
\begin{equation}
  \label{eq:npcomponent}
  \Gamma^\lambda{}_{[a,b]} = (\bA_\xi\cdot \Gamma)^\lambda{}_b
  \Gamma^\xi{}_a -\frac{1}{2}\Upsilon^\lambda{}_{ab},
\end{equation}
where
\begin{align}\label{eq:npcomponent2}
  \Upsilon^\lambda{}_{ab} &:=
  \tR^\lambda{}_{ab}- C^\lambda{}_{{\rho_1}{\rho_2}}\tGamma^{\rho_1}{}_a
  \tGamma^{\rho_2}{}_b + 2 \Gamma^{\lambda}{}_{c\hphantom{[} }
  \tGamma^{\rho_1}{}_{[a}
  A^c{}_{b]{\rho_1}}-
  2\Gamma^{\lambda_1} {}_{a}\tGamma^{\rho_1}{}_{b}C^\lambda{}_{{\lambda_1} {\rho_1}}  \\
  \nonumber &\qquad\qquad+ 2\, \Gamma^{\lambda}{}_{c\hphantom{[} }
  \Gamma^{\lambda_1}{}_{[a} A^c{}_{b]\lambda_1} -
  C^\lambda{}_{\lambda_1\lambda_2}\Gamma^{\lambda_1}{}_a
  \Gamma^{\lambda_2}{}_b,
\end{align}
and where $\lambda_1,\lambda_2=s_1+1,\ldots s_0$ have the same range
as $\lambda$.  We will refer to constants $\tR^\alpha{}_{ab},
\tGamma^\rho{}_a$ and scalars $\Gamma^\lambda{}_a$ together with
constraints \eqref{eq:configabianchi}-\eqref{eq:configdbianchi} and
\eqref{eq:confignp}-\eqref{eq:npcomponent2} as a $\CH_1$
configuration \cite{milpel09}.


It is well known that the $O(\eta)$ metric equivalence problem has
trivial essential torsion \cite[Section 12]{olver}.  However, if we
reduce the structure group to $G_0\subset O(\eta)$ by means of
curvature normalization, we obtain the following reduced first
structure equations:
\[  \d \homega^a = -\sum_{\xi=1}^{s_0} A^a{}_{b\xi} \,\hGamma^\xi
\wedge \homega^b
+ \sum_\rho A^a{}_{b\rho}\hGamma^\rho{}_b\, \homega^a \wedge \homega^b ,\]
where
\[ \hbomega = X \bomega,\quad \hGamma = X\Gamma X^{-1} - \d X
X^{-1},\quad X\in G_0,\] are the $G_0$-lifted 1-forms.  The scalars
$\hGamma^\rho{}_a$ are well defined because the Maurer-Cartan term $\d
X X^{-1}$ takes values in $\fg_0$.  Consequently, the scalars
$\Gamma^\rho{}_a$ have a $G_0$-transformation law that does not depend
on $\d X$, and therefore constitute the essential torsion for the 1st
iteration of the equivalence method.  Thus, the scalars
$\Gamma^\rho{}_a$ can be interpreted as the essential torsion arising
from the reduced $G_0$-equivalence problem and the
$\Gamma^\lambda{}_a$ as essential torsion in the next iteration of the
$G_1$-equivalence problem. Therefore, we refer to the former as 1st
order torsion, and to the latter as 2nd order torsion.

By virtue of \eqref{eq:R1Gamma}, normalizing the $\Gamma^\rho{}_a$
is equivalent to normalizing $R_{abcd;e}$. The 1st order
normalization reduces the structure group to $G_1\subset G_0$.  If
we suppose that the $\CH_1$ property holds, then the resulting
invariants are the constants $\tGamma^\rho{}_a$.  The scalars
$\Gamma^\lambda{}_a$ are the essential torsion of the 2nd iteration
of the equivalence method.  By virtue of \eqref{eq:R1Gamma}
\eqref{eq:R2Gamma}, the 2nd order Cartan invariants are functions of
the 0th order Cartan invariants $R^\alpha{}_{ab}$ and the 1st and
2nd order torsion scalars $\tGamma^\rho{}_a, \Gamma^\lambda{}_a$.
Inversely, because of \eqref{eq:configabianchi}, $\tR^\alpha{}_{ab}$
is linearly dependent on $\tR_{abcd}$, while \eqref{eq:tRabcde} and
\eqref{eq:R2Gamma} can be solved to give $\tGamma^\rho{}_a$ as
functions of $\tR_{abcd},\tR_{abcde}$ and $\Gamma^\lambda{}_a$ as a
function of $\tR_{abcd}, \tR_{abcde}, R_{abcd;ef}$.

\end{document}